\documentclass[11pt]{article}
\pdfoutput=1
\usepackage{soul}
\usepackage{jheppub}
\usepackage{amsfonts}
\usepackage{amsmath}
\usepackage{amsthm}
\allowdisplaybreaks[4]         
\usepackage{amssymb}   
\usepackage{euscript}   
\usepackage{cleveref}    
\usepackage[dvipsnames]{xcolor}          
\usepackage{tensor}     
\usepackage{graphicx}
\usepackage{caption}
\usepackage{subcaption}   
\usepackage{float}
\usepackage{hyperref}
\hypersetup{
	colorlinks =NavyBlue,
	linkcolor =NavyBlue,
	citecolor=MidnightBlue,
	filecolor=NavyBlue,
	urlcolor=NavyBlue,
}

\newcommand{\bea}{\begin{eqnarray}}
\newcommand{\eea}{\end{eqnarray}}
\newcommand{\ba}{\begin{eqnarray}}
\usepackage{braket}
\newcommand{\ea}{\end{eqnarray}}

\newcommand{\beq}{\begin{equation}}
\newcommand{\eeq}{\end{equation} }
\newcommand{\beqa}{\begin{eqnarray}}

\newcommand{\eeqa}{\end{eqnarray}}
\newcommand{\beqar}{\begin{eqnarray*}}
\newcommand{\eeqar}{\end{eqnarray*}}

\newcommand{\be}{\begin{equation}}
\newcommand{\ee}{\end{equation}}
\newcommand{\diff}{\mathrm{d}}

\newtheorem{theorem}{Theorem}

\newtheorem{defi}{Definition}

\newtheorem{proposition}{Proposition}
\newtheorem{example}{Example}

\newcommand{\E}{\mathcal{E}}
\newcommand{\I}{\mathcal{I}}

\newcommand{\sZ}{\mathsf{Z}}
\newcommand{\sW}{\mathsf{W}}
\newcommand{\sX}{\mathsf{X}}
\newcommand{\sY}{\mathsf{Y}}

\newcommand{\cL}{\mathcal{L}}

\newcommand{\ex}[1]{\text{e}^{#1}} 
\newcommand{\iu}{\text{i}}

\definecolor{shadecolor}{rgb}{.25,.25,.25}

\title{ \boldmath Cosmological higher-curvature gravities}

\author[a]{Javier Moreno,}
\author[b]{\'Angel J. Murcia}

\affiliation[a]{Departamento de Física, Universidad de Concepción, Casilla, 160-C, Concepción, Chile.
 \vspace{0.1cm}}
\affiliation[b]{INFN, Sezione di Padova, Via Francesco Marzolo 8, 35131 Padova, Repubblica Italiana \vspace{0.1cm}}

\emailAdd{jmoreno@campus.haifa.ac.il}
\emailAdd{angel.murcia@pd.infn.it}

\date{\today}
\abstract{We examine higher-curvature gravities whose FLRW configurations are specified by equations of motion which are of second order in derivatives, just like in Einstein gravity. We name these theories \emph{Cosmological Gravities} and initiate a systematic exploration in dimensions $D \geq 3$. First, we derive an instance of Cosmological Gravity to all curvature orders and dimensions $D \geq 3$. Second, we study Cosmological Gravities admitting non-hairy generalizations of the Schwarzschild solution characterized by a single function whose equation of motion is, at most, of second order in derivatives. We present explicit instances of such theories for all curvature orders and dimensions $D \geq 4$. Finally, we investigate the equations of motion for cosmological perturbations in the context of generic Cosmological Gravities. Remarkably, we find that the linearized equations of motion for scalar cosmological perturbations in any Cosmological Gravity in $D\geq 3$ contain no more than two time derivatives. We explicitly corroborate this aspect by presenting the equations for the scalar perturbations in some four-dimensional Cosmological Gravities up to fifth order in the curvature.
 }

\begin{document} 
\maketitle
\flushbottom


\section{Introduction}
\label{sec:Introduction}

The quest for the theory of Quantum Gravity remains as one of the most puzzling, challenging and fundamental problems of Physics. It has been a vivid and central theme for both the gravitational and theoretical high-energy physics communities over the last decades, in which some potential candidates for  Quantum Gravity have been put forward and developed, String Theory being the most prominent one. At the level of low-energy effective actions, it corrects the Einstein-Hilbert Lagrangian through the introduction of specific higher-order terms in the spacetime curvature \cite{Callan:1985ia,Zwiebach:1985uq,Grisaru:1986vi,Gross:1986iv,Gross:1986mw,Bergshoeff:1989de}.  This way, the study of quantum corrections to General Relativity (GR) amounts to the exploration of \emph{higher-order gravities} (equivalently, higher-curvature or higher-derivative gravities). 

As a matter of fact, in recent years there has been an ever-growing interest in higher-order gravities by themselves, regardless of their possible fundamental origin. Indeed, higher-curvature terms arise naturally in gravitational effective actions, in which one is required to include all terms compatible with diffeomorphism-invariance \cite{Donoghue:1994dn,Weinberg:1995mt}. Further motivations for their investigation include the fact that higher-derivative terms usually give rise to renormalizable gravitational actions \cite{Stelle:1976gc, Stelle:1977ry} --- so that they capture the relevant physics when the spacetime curvature is sufficiently large (such as in black-hole physics \cite{Wheeler:1985nh,Boulware:1985wk,Wiltshire:1988uq,Myers:2010ru,Lu:2015cqa,Bueno:2016lrh,Bueno:2017sui,Aguilar-Gutierrez:2023kfn}) ---, their numerous holographic applications \cite{Brigante:2007nu,Hofman:2008ar,Myers:2008yi,Cai:2008ph,deBoer:2009pn,Camanho:2009vw,Buchel:2009sk,deBoer:2009gx,Myers:2010jv,Perlmutter:2013gua,Mezei:2014zla,Bueno:2015rda,Bueno:2015xda,Chu:2016tps,Dey:2016pei,Bueno:2018xqc,Li:2018drw,Bueno:2018yzo,Bueno:2020uxs,Bueno:2020odt,Anastasiou:2021swo,Anastasiou:2021jcv,Cano:2022ord,Bueno:2022jbl,Murcia:2023zok} or their capability to parametrize potential deviations from Einstein gravity in astrophysical observables \cite{Cardoso:2009pk,Blazquez-Salcedo:2017txk,Berti:2018cxi,Cardoso:2018ptl,LIGOScientific:2018mvr,LIGOScientific:2020ibl,Cano:2021myl,LIGOScientific:2021djp,Silva:2022srr,Cano:2023jbk,Cayuso:2023xbc}. 

Another realm in which higher-curvature gravities have proven to be extremely useful is that of cosmology. Such an insight is due to Starobinsky \cite{Starobinsky:1980te,Starobinsky:1983zz}, who realized that adding a Ricci-scalar squared term $R^2$ into the Einstein-Hilbert action provides a natural explanation for cosmic inflation, which furthermore fits very well to current observations \cite{Planck:2018jri}. This has motivated the exploration of more generic higher-curvature gravities for cosmological purposes, placing special emphasis on the investigation of $f(R)$ theories \cite{Li:2007xn,Sotiriou:2008rp,Nojiri:2010wj,Clifton:2011jh,Nojiri:2017ncd,Cheong:2020rao,Ivanov:2021chn,Odintsov:2023weg}. Their choice corresponds to the fact that they are classically equivalent to scalar-tensor theories with equations of motion of second order in derivatives \cite{Barrow:1988xh}. Nonetheless, it is clear that these theories do not represent the most general effective theories of gravity, which would feature higher-curvature terms with explicit Riemann curvature and Ricci curvature tensors. 

However, given the vast number of higher-curvature gravities existing at each curvature order, it is more convenient to initiate the (cosmological) exploration of such more generic theories with the study of certain higher-curvature gravities --- including Riemann curvature and Ricci curvature tensors --- possessing some special features, such as having second-order equations of motion for Friedmann-Lemaître–Robertson–Walker (FLRW) configurations. There are several reasons to  impose such a condition. First, because it dramatically simplifies the study of cosmological solutions with higher-curvature corrections. Second, since an analogous requirement in the context of spherically symmetric solutions has turned out to be greatly successful \cite{Oliva:2010eb,Myers:2010ru,Dehghani:2011vu,Bueno:2016xff,Hennigar:2017ego,Bueno:2019ltp,Bueno:2019ycr,Bueno:2022res,Aguilar-Gutierrez:2023kfn}. And third, because theories with second-order equations of motion on cosmological backgrounds automatically provide instances of holographic theories satisfying a holographic $c$-theorem \cite{Freedman:1999gp,Myers:2010xs,Myers:2010tj,Bueno:2022lhf,Bueno:2022log}. 

Interestingly enough, examples of higher-curvature theories in four spacetime dimensions with second-order equations for FLRW backgrounds have been constructed at cubic order in the curvature \cite{Arciniega:2018fxj}, quartic (and quintic) order \cite{Cisterna:2018tgx} and up to eighth order in the curvature \cite{Arciniega:2018tnn}. Among other things, it was found that the aforementioned cubic theory yields equations of motion for scalar perturbations on top of FLRW backgrounds which are of second order in time derivatives \cite{Cisterna:2018tgx} and that  they provide a mechanism for cosmic inflation devoid of additional fields to the metric \cite{Arciniega:2018fxj,Arciniega:2018tnn,Edelstein:2020nhg,Jaime:2021pqn,Jaime:2022cho}. For other aspects of these theories, we refer the reader to \cite{Erices:2019mkd,Arciniega:2020pcy,Quiros:2020eim,Edelstein:2020lgv,Cano:2020oaa,Asimakis:2022mbe} and references therein. Also, in the context of theories fulfilling a holographic $c$-theorem, higher-order gravities with second-order equations on top of FLRW backgrounds were derived at any curvature order in $D=3$ \cite{Bueno:2022lhf}, while for generic dimensions a recursive algorithm to find Cosmological Gravities at all curvature orders (although with explicit covariant derivatives of the curvature) was provided in Reference \cite{Bueno:2022log}.


However, there still remain many open questions regarding theories with second-order equations for the scale factor in FLRW configurations. For example, is it possible to obtain an explicit example of a theory of this class at all curvature orders (and without covariant derivatives) in four dimensions? And in generic dimensions $D \geq 4$? Could it be feasible to derive an instance of a theory at any curvature order and dimensions $D \geq 4$ that possesses second-order equations of motion for both FLRW backgrounds and static and spherically symmetric configurations? And regarding cosmological perturbations, could it be possible that ensuring second-order equations of motion for FLRW backgrounds forces the equations for scalar perturbations to be of second order in time derivatives?

In the present manuscript, we address and answer all previous questions in the affirmative. Specifically, if we coin the name \emph{Cosmological Gravity} to refer to any (higher-curvature) theory of gravity with second-order equations on top of FLRW configurations, we have managed to derive the results listed below.
\subsection{Summary of results}
\begin{enumerate}
\item Section \ref{sec:hogflrw} provides some preparatory material regarding the properties of curvature invariants when evaluated on top of FLRW backgrounds. More concretely, it is proven that the equations of motions for FLRW backgrounds in generic higher-curvature gravities may be obtained through the variation of the Lagrangian evaluated on such ans\"atze and that all curvature invariants on these configurations may be expressed as polynomials of two independent functions.
\item Section \ref{sec:cosmograv} introduces the class of Cosmological Gravities and presents an extremely simple condition (see Proposition \ref{prop:cosmocond}) to determine whether a theory is a Cosmological Gravity. This condition is further used to derive, for the first time, an instance of a Cosmological Gravity at every curvature orders and for all spacetime dimensions\footnote{Clearly, one would be primarily interested in $D=4$. However, the generalization for arbitrary $D$ is straightforward and could be of utility in the context of stringy effective actions and/or holography.} $D \geq 3$, see Theorem \ref{theo:cosmogenteorias}.
\item Section \ref{sec:CGQG} focuses on the interplay between Cosmological Gravities and the class of Generalized Quasitopological Gravities, characterized by possessing non-hairy spherically symmetric black holes fully specified by one single function. On the one hand, for $D \geq 5$ we find an example of a Cosmological Gravity also belonging to the restricted class of Quasitopological Gravities (in which the equation for the function characterizing the spherically symmetric solution is algebraic) at all curvature orders. Then, on the other hand, in $D=4$ we are able to construct an instance of a \emph{Cosmological Generalized Quasitopological Gravity} at all curvature orders.
\item Section \ref{sec:cosmopert} is devoted to the study of cosmological perturbations in Cosmological Gravities. More concretely, scalar perturbations are considered, proving that all Cosmological Gravities in $D \geq 3$ possess equations of motion for the scalar perturbations on top of FLRW backgrounds which are of second order in time derivatives (although featuring higher-order spatial derivatives). We also study the cosmological scalar perturbations for the four-dimensional Cosmological Generalized Quasitopological Gravities derived in Section \ref{sec:CGQG} and show the associated equations of motion for the perturbations up to order $n=5$.

\item Finally, in Section \ref{sec:conclu} we present our conclusions and discuss possible future directions.
\end{enumerate}









\subsubsection*{Note on conventions}

We will be using the spacelike (i.e., mostly-plus) signature for the metric $g_{ab}$, as well as the conventions of Reference \cite{Wald:1984rg} for the spacetime curvature. Specifically, the Riemann curvature tensor $R_{abc}{}^d$ reads: 
\begin{equation}
R_{abc}{}^{d}=-2\partial_{[a} \Gamma_{b]c}{}^{d}+ 2\Gamma_{c[a}{}^{e} \Gamma_{b]e}{}^{d}\,,
\end{equation}
where $\Gamma_{ab}{}^{c}$ denotes the Christoffel symbols and where $[ab]$ stands for antisymmetrization of indices, including the corresponding weights. The Ricci curvature tensor $R_{ab}$ and the Ricci scalar $R$ are defined as follows:
\begin{equation}\label{eq:RicS}
R_{ab}=R_{adb}{}^d\, , \qquad R=R_{b}{}^{b}\,. 
\end{equation}
In another vein, the Weyl tensor $W_{abc}{}^d$ and the traceless Ricci tensor $Z_{ab}$ of a $D$-dimensional Lorentzian manifold are given by:
\begin{align}
W_{abcd}&=R_{abcd}+\frac{2}{D-2}(g_{b[c} R_{d]a}-g_{a[c} R_{d]b})+\frac{2}{(D-1)(D-2)} R g_{c[a}g_{b]d}\,\label{eq:Weyl},\\
Z_{ac}&=R_{ac}-\frac{1}{D}g_{ac} R\,.\label{eq:TraclessR}
\end{align}

\section{Higher-order gravities and their structure on FLRW configurations}\label{sec:hogflrw}

Consider the most general diffeomorphism-invariant theory of gravity $\mathcal{L}(g^{ab},R_{abcd})$ which is built out from arbitrary contractions\footnote{We will not consider parity-breaking terms or terms with covariant derivatives of the curvature.} of the Riemann curvature tensor $R_{abcd}$ associated to a metric $g_{ab}$. We will assume that it reduces to GR at sufficiently low energies and that it admits an effective expansion in powers of the curvature, so that it can be expressed as follows:
\begin{equation}
\mathcal{L}(g^{ab},R_{abcd})=-2\Lambda+ R + \sum_{n=2}^\infty \sum_{m=1}^{g_n} \ell^{2(n-1)} \alpha_{n,m} \mathcal{R}_{(n,m)}\,,
\label{eq:hogans}
\end{equation}
where $\Lambda$ stands for the cosmological constant, $\ell$ denotes the typical length scale from which on the effects of higher-curvature terms become relevant (e.g., the Planck scale), $\alpha_{n,m}$ are dimensionless couplings and the terms $\mathcal{R}_{(n,m)}$ represent the $g_n$ different (linearly independent) curvature invariants which can be constructed at order $2n$ in derivatives\footnote{To clarify this nomenclature, observe that each Riemann curvature tensor introduces two derivatives of the metric, so a term of order $2n$ in derivatives would be formed by $n$ curvature tensors.}. In particular, the Lagrangian \eqref{eq:hogans} takes the form of a higher-order (or higher-curvature) gravity, which generically possesses infinitely-many terms with an increasing number of derivatives, but which could also be truncated at a certain order. 

We shall be interested in studying $D$-dimensional ($D \geq 3$) cosmological configurations corresponding to the following FLRW ansatz for the metric:
\begin{equation}
\diff s_{a,F}^2=- F(t) \diff t^2+ \frac{a^2(t)}{1-k r^2}\diff r^2+  a^2(t) r^2 \mathrm{d} \Omega^2_{D-2}\,,
\label{eq:flrwans}
\end{equation}
where $a(t)$ stands for the scale factor and $k=-1,0,1$ controls the topology of the spatial sections (hyperbolic, planar or spherical). Two remarks should be done at this point. On the one hand, we allow the spacetime dimension $D$ to be an arbitrary integer larger than two. Indeed, higher-curvature gravities in dimensions larger than four appear naturally in the context of String Theory (see e.g. References \cite{Callan:1985ia,Zwiebach:1985uq,Gubser:1998nz,Bergshoeff:1989de}) and within the exploration of holographic Conformal Field Theories (CFTs) in generic dimensions (for instance, check References \cite{Brigante:2007nu,deBoer:2009pn,Camanho:2009vw,Buchel:2009sk,Myers:2010jv,Myers:2010xs,Myers:2010tj,Li:2018drw,Bueno:2022lhf,Bueno:2022log}). On the other hand, observe that we have included a time-dependent function $F(t)$, which could be reabsorbed via a time-reparametrization. The main motivation to keep $F(t)$ explicit is that it will be highly useful for the derivation of the subsequent generalized Friedmann equations (the equations of motion of $\mathcal{L}(g^{ab},R_{abcd})$ on FLRW configurations \eqref{eq:flrwans}), as we will momentarily show. For further reference, let us also present corresponding FLRW metric with $F(t)=1$, which we will indistinctly denote as $\diff s_{a,1}^2$, or $\diff s_{a}^2$:
\begin{equation}
\diff s_{a}^2=- \diff t^2+ \frac{a^2(t)}{1-k r^2}\diff r^2+ a^2(t)  r^2 \mathrm{d} \Omega^2_{D-2}\,,
\label{eq:flrwansf1}
\end{equation}
If necessary, we will refer to $\diff s_{a}^2$ as the comoving FLRW ansatz. Analogously, we will denote by $L_{a,F}= \mathcal{L} \vert_{a,F}$ the reduced Lagrangian obtained by evaluation of $\mathcal{L}(g^{ab},R_{abcd})$ on the FLRW ansatz \eqref{eq:flrwans}.
\begin{proposition}
\label{prop:eomflrw}
Let $\mathcal{L}(g^{ab},R_{abcd})$ be any (higher-order) theory of gravity\footnote{For this proposition, we could also consider covariant derivatives of the curvature.}. The equations of motion of $\mathcal{L}(g^{ab},R_{abcd})$ for FLRW configurations \eqref{eq:flrwans} are given by the Euler-Lagrange variation of $\hat{L}_{a,F}=\left. \left (\sqrt{\vert g \vert}\right ) \right \vert_{a,F} L_{a,F}$ with respect to $a$ and $F$:
\begin{align}
 \frac{\delta \hat{L}_{a,F}}{\delta F} &=\frac{\partial \hat{L}_{a,F}}{\partial F}-  \frac{\diff}{\diff t} \left(\frac{\partial \hat{L}_{a,F}}{\partial \dot{F}} \right)  =0\, , \\ \frac{\delta \hat{L}_{a,F}}{\delta a}  &=\frac{\partial \hat{L}_{a,F}}{\partial a}- \frac{\diff}{\diff t} \left( \frac{\partial \hat{L}_{a,F}}{\partial \dot{a}} \right)+\frac{\diff^2}{\diff t^2}\left( \frac{\partial \hat{L}_{a,F}}{\partial \ddot{a}}\right)=0\,,
\end{align}
where\footnote{We have not included the derivative of $\hat{L}_{a,F}$ with respect to $\ddot{F}$ since it does not appear in any curvature invariant, as we will see afterwards. } $\mathrm{dot}$ stands for time derivative. After the functional variation, one can reparametrize time to set $F=1$.
\end{proposition}
\begin{proof}
Let us consider an arbitrary metric $g_{ab}$ and some coordinates $(t,r,x^{i})$, with $i=1,2,...,D-2$. Write $g_{tt}=-F(t)$, $g_{rr}=a^2(t)/(1-kr^2)$ and $g_{ij}=a^2(t) r^2 h_{ij}$, where the indices $i,j$ refer the $(D-2)$ coordinates transverse to the $(t,r)$ coordinates. Observe that:
\begin{align}
\label{eq:eomgennocos1}
 \frac{\delta \hat{L}_{a,F}}{\delta F}&=\sqrt{\vert g \vert}\mathcal{E}_{ab} \frac{\partial g^{ab}}{\partial F}=\sqrt{\frac{F}{1-k r^2}} a^{(D-1)}\frac{1}{F^2} \mathcal{E}_{tt}\,,\\
 \label{eq:eomgennocos2}
 \frac{\delta \hat{L}_{a,F}}{\delta a}&=\sqrt{\vert g \vert} \mathcal{E}_{cd} \frac{\partial g^{cd}}{\partial a}=\sqrt{\frac{F}{1-k r^2}} a^{(D-1)}\left(-\frac{2(1-kr^2)}{ a^3} \mathcal{E}_{rr}-\frac{2}{r^2 a^3} h^{ij} \mathcal{E}_{ij}\right) \,,
\end{align}
where $\mathcal{E}_{ab}$ denotes the generic gravitational field equations of the theory $\mathcal{L}(g^{ab},R_{abcd})$ (not evaluated on any particular configuration). If we now evaluate $\mathcal{E}_{ab}$ on the FLRW ansatz \eqref{eq:flrwans}, so that $t$ becomes naturally a time coordinate, $r$ represents a radial variable and $h_{ij}$ is simply the metric of the $(D-2)$-dimensional round sphere, then the symmetries of such an FLRW spacetime fix $\mathcal{E}_{ab}$ to be of the form:
\begin{equation}
\mathcal{E}_{ab}=(s_1(t)+F(t) s_2(t)) \delta_a^t \delta_b^t+s_2(t) g_{ab}\,,
\label{eq:geneomcos}
\end{equation}
where $g_{ab}$ is the FLRW metric given in Equation \eqref{eq:flrwans}. This implies that the whole system of Einstein equations reduces to a system of two coupled ordinary differential equations. At this point, evaluation of Equations \eqref{eq:eomgennocos1} and \eqref{eq:eomgennocos2} on the background \eqref{eq:flrwansf1} and subsequent use of Equation \eqref{eq:geneomcos} reveals that:
\begin{align}
\label{eq:eomgensicos1}
\left. \frac{\delta \hat{L}_{a,F}}{\delta F} \right \vert_{\diff s^2_{a,F}}&=\sqrt{\frac{F}{1-k r^2}} a^{(D-1)}\frac{1}{F^2} s_1(t)\,,\\
\left. \frac{\delta \hat{L}_{a,F}}{\delta a} \right \vert_{\diff s^2_{a,F}}&=-2\sqrt{\frac{F}{1-k r^2}} a^{(D-4)} (D-1)s_2(t)\,.
\label{eq:eomgensicos2}
\end{align}
Since $s_1(t)$ and $s_2(t)$ are precisely the unique independent components of the gravitational equations of motion, we conclude.
\end{proof}

Proposition \ref{prop:eomflrw} provides us with a very convenient way to compute the equations of motion of any theory for (comoving) FLRW ans\"atze \eqref{eq:flrwansf1}. In particular, one just needs to evaluate the corresponding higher-curvature Lagrangian $\mathcal{L}(g^{ab},R_{abcd})$ on the configuration \eqref{eq:flrwans} and carry out the functional derivatives of the resulting reduced Lagrangian with respect to $a$ and $F$ (and setting afterwards $F(t)=1$ if desired). This approach corresponds to the so-called principle of symmetric criticality \cite{Palais:1979rca}, which may be considered as a classical counterpart of the celebrated mini-superspace formalism in the context of Loop Quantum Gravity \cite{DeWitt:1967yk,Dewitt:1968lxx}. 

It would be extremely convenient for the study of cosmological solutions of higher-order gravities to have a better understanding on the structure of higher-curvature terms when evaluated on the FLRW ansatz \eqref{eq:flrwans}. To this aim, it is appropriate to express all curvature invariants in terms of  Weyl tensors $W_{abcd}$ \eqref{eq:Weyl}, traceless Ricci tensors $Z_{ab}$ \eqref{eq:TraclessR} and Ricci scalars $R$ \eqref{eq:RicS}. This is clearly useful, since metrics of the form given at Equation \eqref{eq:flrwans} are conformally flat:
\begin{equation}
\left. W_{abcd}  \right \vert_{\diff s^2_{a,F}}=0\,.
\end{equation}
Hence any higher-curvature term containing Weyl tensors will exactly vanish on FLRW configurations, so we may forget about them for the computation of the subsequent equations of motion on such backgrounds. We need, thus, to concentrate on those terms made up with traceless Ricci tensors and Ricci scalars. For that, we observe that:
\begin{equation}
\label{eq:gammadef}
\left.Z_a^b\right\vert_{\diff s^2_{a,F}}=\Gamma(-D \delta_a^t \delta_t^b+\delta_a^b)\,, \quad \Gamma=\frac{(D-2)(2k F^2+2 F(\dot{a})^2+a \dot{a} \dot{F}-2a \ddot{a} F)}{2 a^2 F^2 D}\,,
\end{equation}
where we remind that \emph{dot} stands for time derivative (with respect to comoving time). This means that all contractions of traceless Ricci tensors will be proportional to an appropriate power of $\Gamma$. Define:
\begin{equation}
\label{eq:sigmadef}
\Sigma=R\vert_{\diff s^2_{a,F}}=\frac{(D-1)\left((D-2)k F^2+(D-2)F(\dot{a})^2-a\dot{a} \dot{F}+2a\ddot{a}F\right)}{a^2F^2}\, .
\end{equation}
The previous argumentation shows the following proposition.
\begin{proposition}\label{prop:noWeyl}
All non-zero curvature invariants built out from $p$ Weyl tensors, $q$ traceless Ricci tensors and $m$ Ricci scalars are proportional to each other when evaluated on the FLRW ansatz \eqref{eq:flrwans}. More concretely, if $(R^m W^p Z^q)_i$ is any possible way of contracting all such tensors, then
\begin{equation}
\label{eq:1propflrw}
(R^m W^p Z^q)_i\vert_{a,F}=0\,,  \quad p \geq 1 \quad \text{and} \quad \left.(R^m Z^p)_i\right\vert_{a,F}=c_i \Sigma^m \Gamma^p\,, \quad c_i \in \mathbb{R}\, ,
\end{equation}
\end{proposition}

This latter proposition implies that any linear combination of invariants of $n$-th order in the curvature, when evaluated on the FLRW ansatz \eqref{eq:flrwans}, takes the form:
\begin{equation}
\left. \mathcal{R}_{(n)}\right \vert_{ \diff s^2_{a,F}}=\sum_{i=1}^p \beta_i \Sigma^{b_1^{(i)}} \Gamma^{b_2^{(i)}}\,,
\label{eq:cosmoanyinv}
\end{equation}
for some coefficients $\beta_i$ and non-negative integers $b_1^{(i)}$, $b_2^{(i)}$ (such that $b_1^{(i)}+b_2^{(i)}=n$) and $p$. As far as one is interested in FLRW configurations, it turns out that one can always exchange any curvature invariant $\mathcal{R}_{(n)}$ by a specific  combination of curvature invariants which, evaluated on an FLRW ansatz, is equal to $\left. \mathcal{R}_{(n)}\right \vert_{ \diff s^2_{a,F}}$. Indeed, since $\left. \mathcal{R}_{(n)}\right \vert_{ \diff s^2_{a,F}}$ is a polynomial of $\Sigma$ and $\Gamma$, one may wonder if it is possible to choose a finite number of low-order curvature invariants which, evaluated on FLRW ans\"atze, generate in a multiplicative way (together with the Ricci scalar) all possible polynomials \eqref{eq:cosmoanyinv}. To this aim, define:
\begin{equation}
\label{eq:kcosmodens}
\mathcal{K}_2=\frac{1}{D(D-1)}Z_{ab} Z^{ab}\, , \quad \mathcal{K}_3=-\frac{1}{D(D-1)(D-2)}Z_{ab} Z^{bc} Z_c^a\,.
\end{equation}
They satisfy:
\begin{equation}
\label{eq:kscosmo}
\left. \mathcal{K}_2 \right \vert_{\diff s^2_{a,F}}= \Gamma^2\, , \quad \left. \mathcal{K}_3 \right \vert_{\diff s^2_{a,F}}= \Gamma^3\,.
\end{equation}
We have the following result.
\begin{proposition}
\label{prop:invcosmo}
Let $\mathcal{R}_{(n)}$ be any linear combination of invariants of $n$-th order in the curvature. Then:
\begin{equation}
\left. \mathcal{R}_{(n)}\right \vert_{ \diff s^2_{a,F}}=\sum_{i=1}^p \beta_i \left. R^{b_1^{(i)}}\right \vert_{ \diff s^2_{a,F}} \left.  \mathcal{K}_2^{b_2^{(i)}} \right \vert_{ \diff s^2_{a,F}}  \left.\mathcal{K}_3^{b_3^{(i)}}\right \vert_{ \diff s^2_{a,F}} \, ,
\end{equation}
for certain coefficients $\beta_i$ and non-negative integers $b_1^{(i)}$, $b_2^{(i)}$, $b_3^{(i)}$ and $p$.
\end{proposition}
\begin{proof}
It follows directly from Proposition \ref{prop:noWeyl} and Equation \eqref{eq:kscosmo}. In another vein, observe that to generate terms with either an even or odd number of traceless Ricci tensors, one needs both curvature invariants $\mathcal{K}_2$ and $\mathcal{K}_3$ presented in Equation \eqref{eq:kcosmodens}.
\end{proof}
This proposition will be of use in the following section.
\section{Cosmological Gravities} \label{sec:cosmograv}

In this section we devote ourselves to the study of particular higher-curvature gravities which possess very special properties on FLRW configurations, which we shall call \emph{Cosmological Gravities}.
\begin{defi}
A theory $\mathcal{L}(g^{ab},R_{abcd})$ is said to be a Cosmological Gravity if its equations of motion for FLRW configurations \eqref{eq:flrwans} --- and henceforth for comoving FLRW configurations \eqref{eq:flrwansf1} --- are of second order in time derivatives.
\end{defi}
The most direct examples of Cosmological Gravities are Einstein Gravity or Lovelock gravities. However, these cases are somewhat trivial, since such theories possess second order equations of motion for any background. There also exist other instances of Cosmological Gravities which do not possess second order equations of motion when evaluated on generic backgrounds, as the following example shows.
\begin{example}
\label{ep:aej}
Consider the following higher-curvature theory in $D=4$ \cite{Arciniega:2018fxj}:
\begin{align}
\label{eq:AEJ}
\mathcal{L}_{\mathrm{AEJ}}&=R-2 \Lambda+ \alpha \ell^4 (\mathcal{P}-8 \hat{\mathcal{P}})\,,\\
\label{eq:paejdef}
\mathcal{P}&=12 R_{a\ b}^{\ c \ d}R_{c\ d}^{\ e \ f}R_{e\ f}^{\ a \ b}+R_{ab}^{cd}R_{cd}^{ef}R_{ef}^{ab}-12R_{abcd}R^{ac}R^{bd}+8R_{a}^{b}R_{b}^{c}R_{c}^{a}\,,\\
\label{eq:phataejdef}
\hat{\mathcal{P}}&=R_{abcd}R^{abc}{}_e R^{de}-\frac 14 R_{abcd} R^{abcd}R-2R_{abcd} R^{ac} R^{bd}+\frac{1}{2}R_{ab} R^{ab} R\,.
\end{align}
Although the generic gravitational field equations are of fourth order in derivatives, this theory possesses second-order equations of motion on top of FLRW backgrounds, so it is an example of a Cosmological Gravity.
\end{example}


In general, given a certain higher-order gravity $\mathcal{L}(g^{ab}, R_{abcd})$, one could check whether it is a Cosmological Gravity by studying whether the whole set of gravitational equations of motion for FLRW backgrounds boils down to the resolution of second-order equations of motion for the scale factor. Nevertheless, it would be desirable to test this condition by a more direct procedure which circumvents the task of deriving the general equations of motion. If $L_{a,F}$ denotes (as before) the reduced Lagrangian obtained by evaluating $\mathcal{L}(g^{ab}, R_{abcd})$ on the FLRW ansatz \eqref{eq:flrwans}, it turns out that there exists a concise requirement to discern whether a given theory is of the cosmological type, as we present next. This condition had already been analyzed in Reference \cite{Paulos:2010ke,Bueno:2022log}, and here we prove in detail that it is a necessary and sufficient condition to have second-order equations for $a(t)$. 
\begin{proposition} 
\label{prop:cosmocond}
A theory $\mathcal{L}(g^{ab},R_{abcd})$ is a Cosmological Gravity if and only if the associated reduced Lagrangian $L_{a,1}=L_{a,F} \vert_{F=1}$ on the comoving FLRW ansatz \eqref{eq:flrwansf1} is linear in $\ddot{a}$:
\begin{equation}
 \frac{\partial^2  L_{a,1} }{\partial \ddot{a}^2} =0\,.
\label{eq:cosmoteo}
\end{equation}
\end{proposition}
\begin{proof}
Proposition \ref{prop:eomflrw} states that the equations of motion of $\mathcal{L}(g^{ab},R_{abcd})$ for FLRW configurations may be computed by working out the functional derivatives of $\hat{L}_{a,F}=\left. \left (\sqrt{\vert g \vert}\right ) \right \vert_{a,F} L_{a,F}$ with respect to $a$ and $F$. Since $\left. \left (\sqrt{\vert g \vert}\right ) \right \vert_{a,F}$ does not feature any derivatives of $a(t)$ and $F(t)$, for the purposes of this proposition we may equivalently work directly with $L_{a,F}$. All potential terms including third or fourth time-derivatives may arise from the following terms, appearing in the aforementioned functional derivatives (note that no second derivatives of $F$ appear in curvature invariants):
\begin{equation}
t_1=\frac{\diff}{\diff t} \left ( \frac{\partial  L_{a,F} }{\partial \dot{F}} \right )\,, \quad t_2=\frac{\diff}{\diff t} \left (\frac{\partial L_{a,F}}{\partial \dot{a}}-\frac{\diff}{\diff t} \left (\frac{\partial L_{a,F}}{\partial \ddot{a}} \right)  \right) \,.
\end{equation}
First, let us take a look at the term $t_1$. It can be proven that it will not produce terms with three time derivatives of the scale factor $a$ if:  
\begin{equation}
\frac{\partial^2  L_{a,F} }{\partial \ddot{a}^2}=0\,.
 \label{eq:t1acero}
\end{equation}
Now, Equation \eqref{eq:cosmoanyinv} implies that the reduced Lagrangian $L_{a,F}$, which we assume includes terms up to $n$-th curvature order, may be expressed as $L_{a,F}=\sum_{p=0}^n \beta_p \Gamma^p \Sigma^{n-p}$, for some couplings $\beta_p$. Then:
\begin{equation}
 \frac{\partial^2  L_{a,F} }{\partial \ddot{a}^2}=\sum_{p=0}^{n-2} \left (2(D-1)\alpha_p(n-1-p) -\frac{D-2}{D}\alpha_{p+1}(p+1)  \right )  \frac{\Gamma^p \Sigma^{n-p-2}}{a^2 F^2}\,,
 \end{equation}
 where we defined $\alpha_p=2(D-1)(n-p) \beta_p-\frac{D-2}{D}(p+1)\beta_{p+1}$. Since the set of $\Gamma^p \Sigma^l$ is linearly independent (both when $F=1$ and when not), we conclude that the previous expression --- and \emph{a fortiori} $ \dfrac{\partial^2  L_{a,F} }{\partial \ddot{a}^2} $ --- is zero if and only if:
\begin{equation}
2(D-1)\alpha_p(n-1-p)=\frac{D-2}{D}\alpha_{p+1}(p+1) \,, \quad p=0,\dots,n-2\,.
\label{eq:conscosmo}
\end{equation}
Let us move now to the study of the term $t_2$. Observe that Equation \eqref{eq:t1acero} guarantees that $t_2$ will possess no more than two time derivatives of $F$, so we can concentrate our efforts on potential higher-derivatives of the scale factor. To this aim, let us momentarily set $F=1$. In such a case, $t_2$ will not contain terms of third order or higher in derivatives of the scale factor if and only if:
\begin{equation}
\frac{\partial}{\partial \ddot{a}} \left ( \frac{\partial L_{a,1}}{\partial \dot{a}}-\frac{\diff}{\diff t } \left ( \frac{\partial L_{a,1}}{\partial \ddot{a}} \right) \right)=0\,.
\label{eq:condeqt2}
\end{equation}
If $\Sigma_{a,1}$ and $\Gamma_{a,1}$ denote evaluation of $\Sigma$ and $\Gamma$ on $F=1$ and defining $\gamma_p=\frac{D-2}{D} \alpha_{p+1}(p+1) +2(D-1)^2(n-p-1) \alpha_p$, after some algebra one finds that:
\begin{equation}
\frac{\partial}{\partial \ddot{a}}\left (\frac{\diff}{\diff t } \left ( \frac{\partial L_{a,1}}{\partial \ddot{a}} \right) \right)=\sum_{p=0}^{n-2} \gamma_p \frac{\dot{a}}{a^3} \Gamma_{a,1}^p \Sigma^{n-p-2}_{a,1}\,.
\label{eq:finalt2}
\end{equation}
On the other hand, the first term of \eqref{eq:condeqt2} reads:
\begin{equation}
\frac{\partial^2 L_{a,1}}{\partial \dot{a} \,  \partial \ddot{a}}=\sum_{p=0}^{n-2}\left ( \frac{2(D-2)}{D}(p+1) \alpha_{p+1}+2(D-1)(D-2) (n-p-1) \alpha_p\right ) \frac{\dot{a}}{a^3} \Gamma_{a,1}^p \Sigma^{n-p-2}_{a,1}\,.
\end{equation}
Subtracting the latter expression to the previous one \eqref{eq:finalt2}, direct use of Equation \eqref{eq:conscosmo} reveals that Expression \eqref{eq:condeqt2} is identically satisfied. If we now recover a generic $F(t)$ by a time reparametrization, we observe that Equation \eqref{eq:condeqt2} would remain unaffected, on account of Equation \eqref{eq:t1acero}. Hence Equation \eqref{eq:t1acero} (either for generic $F(t)$ or for $F=1$) is a necessary and sufficient condition for a theory to be a Cosmological Gravity and we conclude. 

\end{proof}

Having derived a convenient way with which to check if a higher-curvature gravity is of the cosmological type, our next goal will be to obtain the form of the equations of motion on FLRW backgrounds for the most general Cosmological Gravity. To this aim, we introduce some notation first.
\begin{defi}
Two Cosmological Gravities are said to be inequivalent if their equations of motion on the FLRW ansatz \eqref{eq:flrwans} are linearly independent. Otherwise, we call them equivalent.
\end{defi}
The previous definition introduces naturally the notion of equivalence classes of Cosmological Gravities. In particular, we will refer to each different equivalence class as \emph{inequivalent} Cosmological Gravity. Those Cosmological Gravities which are equivalent to the theory defined by a Lagrangian equal to zero form the trivial class of Cosmological Gravities, and the remaining theories will give rise to non-trivial inequivalent Cosmological Gravities. We are going to see next that there is one and only one non-trivial inequivalent Cosmological Gravity at each curvature order and all dimensions $D \geq 3$. 


For that, consider a higher-curvature gravity $\mathcal{L}=-2\Lambda +R+ \sum_{n=2}^\infty \mathcal{L}^{(n)}$, each piece $\mathcal{L}^{(n)}$ being of $n$-th curvature order. By Proposition \ref{prop:noWeyl} and Equation \eqref{eq:cosmoanyinv}, the evaluation of  $\mathcal{L}^{(n)}$ on the comoving FLRW ansatz \eqref{eq:flrwansf1}, which we denote as $L_{a,1}^{(n)}$, reads:
\begin{equation}
L_{a,1}^{(n)}=\sum_{p=0}^n \beta_n \Gamma_{a,1}^n \Sigma_{a,1}^{n-p}\, , \quad \beta_1=0\,,
\end{equation}
where we assume that both $\Gamma$ and $\Sigma$ are evaluated at $F=1$. Also, $\beta_1=0$ since no linear piece in $\Gamma$ may come from an actual curvature invariant evaluated on a comoving FLRW configuration. The condition \eqref{eq:cosmoteo}, in terms of the couplings $\beta_n$, was obtained in Equation \eqref{eq:conscosmo}, where it was written in terms of $\alpha_p=2(D-1)(n-p) \beta_p-\frac{D-2}{D}(p+1)\beta_{p+1}$. Equation \eqref{eq:conscosmo} may be expressed in terms of the constants $\beta_n$ as follows:
\begin{equation}
\begin{split}
\frac{(D-2)^2}{2D^2}(p+1)(p+2) \beta_{p+2} &-2 \frac{(D-1)(D-2)}{D}(p+1)(n-p-1) \beta_{p+1} \\& +2(D-1)^2 (n-p)(n-p-1) \beta_p=0\,, \quad p=0,..., n-2\,.
\end{split}
\label{eq:sisconalfa}
\end{equation}
If we set $\beta_0$ to be a free parameter, then \eqref{eq:sisconalfa} defines a linear system of $n-2$ equations for $n-2$ unknowns, since $\beta_1=0$. For every $n$, it turns out that it admits a unique solution given by:
\begin{equation}
\beta_p=-2^p(p-1) \binom{n}{p} \left ( \frac{D(D-1)}{D-2} \right)^p \beta_0\,.
\end{equation}
Therefore, we conclude that there exists a unique equivalence class of Cosmological Gravities at each order $n$, since the subsequent equations of motion on FLRW backgrounds, up to a constant, are unambiguously fixed at each curvature order. In virtue of Proposition \ref{prop:invcosmo} and Equation \eqref{eq:kscosmo}, it turns out that a representative $\mathcal{C}^{(n)}$ for each equivalence class may be expressed as follows\footnote{These $n$-th order invariants $\mathcal{C}^{(n)}$ turn out to satisfy the following intriguing recursive relation:
\begin{equation}
(n-3)\mathcal{C}^{(n)}=2(n-2)R\, \mathcal{C}^{(n-1)}-(n-1)\mathcal{C}^{(2)}\mathcal{C}^{(n-2)}\,.
\end{equation}}:
\begin{equation}
\begin{split}
\mathcal{C}^{(n)}&= R^n - \sum_{p=1}^{\lfloor n/2 \rfloor}  (2p-1)  \binom{n}{2p} \left ( \frac{D(D-1)}{D-2} \right)^{2p} R^{n-2p} (4 \mathcal{K}_2)^p\\&-\sum_{p=1}^{\lfloor (n-1)/2 \rfloor} 16 p  \binom{n}{2p+1} \left ( \frac{D(D-1)}{D-2} \right)^{2p+1}  R^{n-2p-1} (4 \mathcal{K}_2)^{p-1} \mathcal{K}_3 \,.
\end{split}
\label{eq:cosmolag}
\end{equation}
This proves the following theorem.
\begin{theorem}
\label{theo:cosmogenteorias}
The most general Cosmological Gravity $\mathcal{L}_{\text{cosmo}}$, up to equivalent Cosmological Gravities, is given by\footnote{Observe that $\mathcal{C}^{(1)}=R$.}
\begin{equation}
\mathcal{L}_{\text{cosmo}}=-2\Lambda+\sum_{p=1}^\infty \mu_p \ell^{2p-2} \mathcal{C}^{(p)},
\label{eq:mostgencosmo}
\end{equation}
where the terms $\mathcal{C}^{(n)}$ were given previously \eqref{eq:cosmolag}, $\mu_p$ are dimensionless couplings and $\ell$ is a length scale.
\end{theorem}
Therefore, we have managed to identify a representative of the unique non-trivial equivalence class of Cosmological Gravities at arbitrary curvature order $n$ and for any space-time dimension $D \geq 3$. Taking into account that Cosmological gravities satisfy automatically a holographic $c$-theorem and vice versa \cite{Bueno:2022log}, instances of Cosmological Gravities were built up to all orders in\footnote{We have checked that, up to an innocent global factor, our Cosmological Gravities \eqref{eq:mostgencosmo} for $D=3$ match exactly with those  theories satisfying a holographic $c$-theorem presented in  \cite{Bueno:2022lhf}, see Equation (141) of that Reference (v2 in arXiv). As pointed there, the $n$-th order three-dimensional invariants coincide, up to terms vanishing on FLRW ans\"atze, with those invariants arising in the expansion of the so-called Born-Infeld Massive Gravity \cite{Gullu:2010pc,Gullu:2010st,Alkac:2018whk}.} $D=3$ \cite{Bueno:2022lhf}, to eighth order in the curvature in four dimensions \cite{Arciniega:2018fxj,Cisterna:2018tgx,Arciniega:2018tnn} and also a recursive  algorithm to construct Cosmological Gravities (with covariant derivatives of the curvature, though) in generic spacetime dimension was provided in Reference \cite{Bueno:2022log}. Now, with the present result, we are able to confirm the existence of Cosmological Gravities (without covariant derivatives of the curvature) by an explicit example in all spacetime dimensions and curvature orders $n$. And, by the argumentation of \cite{Bueno:2022log}, we can equivalently claim to have found an instance of a theory satisfying a holographic $c$-theorem for all dimensions and curvature orders.

The lowest-order Cosmological Gravities presented in Equation \eqref{eq:cosmolag} read as follows\footnote{If we set $D=3$ in $\mathcal{C}^{(2)}$, we get the well-known New Massive Gravity theory \cite{Bergshoeff:2009aq,Bergshoeff:2009hq}. Also for $D=3$, our $\mathcal{C}^{(3)}$ and $\mathcal{C}^{(4)}$ exactly match with those theories presented in \cite{Alkac:2020zhg} and obtained as the ``holographic limit'' of Lovelock gravities.}:
\begin{align}
\mathcal{C}^{(1)}&=R\, ,\\
\mathcal{C}^{(2)}&=-\frac{4 D(D-1)}{(D-2)^2} Z_{ab}Z^{ab}+R^2\, ,\\
\mathcal{C}^{(3)}&=\frac{16 (D-1)^2 D^2 Z^a_b Z^b_cZ_a^c}{(D-2)^4}-\frac{12 (D-1) D R Z^a_b Z_a^b}{(D-2)^2}+R^3\, ,\\
\mathcal{C}^{(4)}&=\frac{64 (D-1)^2 D^2 R Z^a_b Z^b_cZ_a^c}{(D-2)^4}-\frac{48 (D-1)^2 D^2 \left(Z^a_b Z_a^b\right)^2}{(D-2)^4}-\frac{24 (D-1) D R^2 Z^a_b Z_a^b}{(D-2)^2}+R^4\, ,\\
\mathcal{C}^{(5)}&=\frac{128 (D-1)^3 D^3 Z^a_b Z_a^bZ^c_d Z^d_eZ_c^e}{(D-2)^6}+\frac{160 (D-1)^2 D^2 R^2 Z^a_b Z^b_cZ_a^c}{(D-2)^4}-\frac{240 (D-1)^2
   D^2 R \left(Z^a_b Z_a^b\right)^2}{(D-2)^4}\notag\\
&-\frac{40 (D-1) D R^3 Z^a_b Z_a^b}{(D-2)^2}+R^5\, ,\\
\mathcal{C}^{(6)}&=\frac{768 (D-1)^3 D^3 R Z^a_b Z_a^bZ^c_d Z^d_eZ_c^e}{(D-2)^6}-\frac{320 (D-1)^3 D^3 \left(Z^a_b Z_a^b\right)^3}{(D-2)^6}+\frac{320 (D-1)^2
   D^2 R^3 Z^a_b Z^b_cZ_a^c}{(D-2)^4}\notag\\
   &-\frac{720 (D-1)^2 D^2 R^2 \left(Z^a_b Z_a^b\right)^2}{(D-2)^4}-\frac{60 (D-1) D R^4
   Z^a_b Z_a^b}{(D-2)^2}+R^6\, .
\end{align}

Let us now present the form of the equations of motion of the most general Cosmological Gravity on comoving FLRW backgrounds \eqref{eq:flrwansf1}. For that, it is enough to consider the theory \eqref{eq:mostgencosmo}. In presence of a perfect fluid with density $\rho$ and pressure $p$, the generalized Friedmann equations associated to the theory \eqref{eq:mostgencosmo} are given by\footnote{Being first-order in time derivatives, \eqref{eq:genfried1} is naturally identified with the Hamiltonian constraint associated with the initial-value problem of our Cosmological Gravities on top of FLRW backgrounds.}:
\begin{align}
\label{eq:genfried1}
\frac{1}{2D} \sum_{n=1}^\infty  \mu_n \ell^{2n-2} \mathcal{F}_n=8 \pi G \rho+\Lambda\,,\\
\label{eq:genfried2}
-\frac{a}{2 D(D-1)\dot{a}}\sum_{n=1}^\infty  \mu_n \ell^{2n-2} \dot{\mathcal{F}}_n=8 \pi G(\rho+P)\,,
\end{align}
where $\Lambda$ is the cosmological constant, $\mu_1=1$ and 
\begin{equation}
\mathcal{F}_n=(D-2n) \left (\frac{D(D-1)(k+\dot{a}^2)}{a^2} \right)^n\,.
\end{equation}
Observe that differentiating \eqref{eq:genfried1} and making use of \eqref{eq:genfried2}, one gets the conservation equation:
\begin{equation}
\frac{\diff \rho}{\diff t}+(D-1) \frac{\dot{a}}{a}(\rho+p)=0\,.
\label{eq:cons}
\end{equation}

\subsection{Numerical resolution of the generalized Friedmann equations}

As an application, we are going to show an efficient way to solve numerically the generalized Friedmann equations \eqref{eq:genfried1} and \eqref{eq:genfried2} for our (four-dimensional) Universe. For that, we will assume that the spatial sections of the Universe are flat ($k=0$) and that it is filled with some matter and radiation. With the additional presence of a non-vanishing cosmological constant, we rewrite the generalized Friedmann equations as follows:
\begin{align}
\label{eq:genfried1new}
\frac{H^2}{H_0^2}+ \sum_{n=2}^\infty \mu_n (\ell H_0)^{2n-2} 12^{n-1} (2-n) \left ( \frac{H}{H_0} \right)^{2n}&= \frac{\Omega_m}{a^3}+\frac{\Omega_r}{a^4}+\Omega_\Lambda\,,\\
-\frac{2 \dot{H}}{3 H_0^2}\left (1+\sum_{n=2}^\infty \mu_n (\ell H_0)^{2n-2} 12^{n-1} n (2-n)\left ( \frac{H}{H_0} \right)^{2n-2} \right)&= \frac{\Omega_m}{a^3}+\frac{4}{3}\frac{\Omega_r}{a^4}\,,
\label{eq:genfried2new}
\end{align}
where we have introduced the Hubble factor $a H=\dot{a}$, with $H_0$ denoting its value at present day, while $\Omega_m$, $\Omega_r$ and $\Omega_\Lambda$ are defined as usual:
\begin{equation}
\Omega_m=\frac{8 \pi G \rho_{m}^{(0)}}{3 H_0^2}\,, \quad \Omega_r=\frac{8 \pi G \rho_{r}^{(0)}}{3 H_0^2}\,, \quad \Omega_\Lambda=\frac{\Lambda}{3 H_0^2}\,,
\end{equation}
where $\rho_m^{(0)}$ and $\rho_r^{(0)}$ denote the current matter and radiation densities in the Universe. In addition to the measurement of the parameters $H_0$, $\Omega_m$, $\Omega_r$ and $\Omega_{\Lambda}$, whose values can be found in \cite{Planck:2018vyg,DESI:2024mwx}, the so-called deceleration  parameter $q_0$, the cosmological jerk $j_0$ and snap $s_0$ are also known. These are defined in terms of the following expansion of the scale factor \cite{Visser:2003vq} around present day:
\begin{equation}
a(t)=1+H_0 t-\frac{H_0^2 q_0}{2} t^2+\frac{H_0^3 j_0}{6} t^3+ \frac{H_0^4 s_0}{24}t^4+\dots\,.
\label{eq:cosmographic}
\end{equation}
Choosing our length scale $\ell$ to be $\ell=\sqrt{\kappa}=\sqrt{8 \pi G}$, substituting \eqref{eq:cosmographic} into \eqref{eq:genfried1new} and solving order by order (up to $t^3$) the resulting equations, one obtains the values of the couplings\footnote{Remember that there is no non-trivial Cosmological Gravity in $D=4$ of quadratic order in the curvature.} $\mu_3,\mu_4,\mu_5$ and $\mu_6$. This is quite interesting, since this provides a direct way with which to experimentally set the values of the higher-curvature couplings associated with our Cosmological Gravities. We use the following values for the cosmological parameters, extracted from the combined data of the Planck and DESI collaborations\footnote{We take the Hubble constant $H_0 \simeq 67.97 \pm 0.38 \, \mathrm{km} \, \mathrm{s}^{-1} \, \mathrm{Mpc}^{-1}$, which is the value obtained when assuming the flat $\Lambda$CDM model for cosmology.} \cite{Planck:2018vyg,DESI:2024mwx}:
\begin{align}
\nonumber
\Omega_m & \simeq 0.3069\,, \quad \Omega_\Lambda \simeq 0.6931\,, \quad \Omega_r \simeq 9.2 \times 10^{-5}\,, \\
 q_0 & \simeq -0.540\,, \quad j_0 \simeq 1.000 \,, \quad s_0 \simeq -0.381\,,
\end{align}
one gets:
\begin{align}
\nonumber
\mu_3&\simeq \frac{5.77 \times 10^{-5}}{\kappa^2 H_0^4 }\,, \quad \mu_4\simeq -\frac{6.00 \times 10^{-6} }{\kappa^3 H_0^6}\,, \\
\label{eq:couplings} \mu_5 &\simeq \frac{2.75\times 10^{-7}}{\kappa^4 H_0^8}\,, \quad \mu_6 \simeq -\frac{4.83\times 10^{-9}}{\kappa^5 H_0^{10}}\,.
\end{align}
Once we have obtained the values of the couplings $\mu_n$ up to sixth order in the curvature, we may proceed to solve the generalized Friedmann equations up to that order, discarding those theories of higher-order in the curvature\footnote{If it were possible to go further in the cosmographic expansion \eqref{eq:cosmographic}, further couplings $\mu_n$ with $n >6$ could be obtained}. To this aim, observe that \eqref{eq:genfried1new} and \eqref{eq:genfried2new} are equivalent to each other after taking into account the equation of state of the different species which are present in the Universe. Indeed, one may get \eqref{eq:genfried2new} by differentiating \eqref{eq:genfried1new} and using the conservation equation \eqref{eq:cons}. Therefore, we can actually use any of the Friedmann equations to solve for the scale factor $a(t)$.

Nevertheless, for our purposes it is more convenient to use \eqref{eq:genfried2new}. Although it is of second order in derivatives, the term $\ddot{a}$ appears linearly in the equation and this greatly simplifies its numerical resolution (instead, \eqref{eq:genfried1new} possess powers of 
twelfth order in $\dot{a}$, hampering the corresponding numerical resolution). Solving at present time $t_0=0$ and using the initial values $a_0=1$ and $\dot{a}_0=H_0$, one gets the scale factor $a(t)$ presented in Figure \ref{fig:1}.

\begin{figure}[H]
\centering
\includegraphics[scale=0.55]{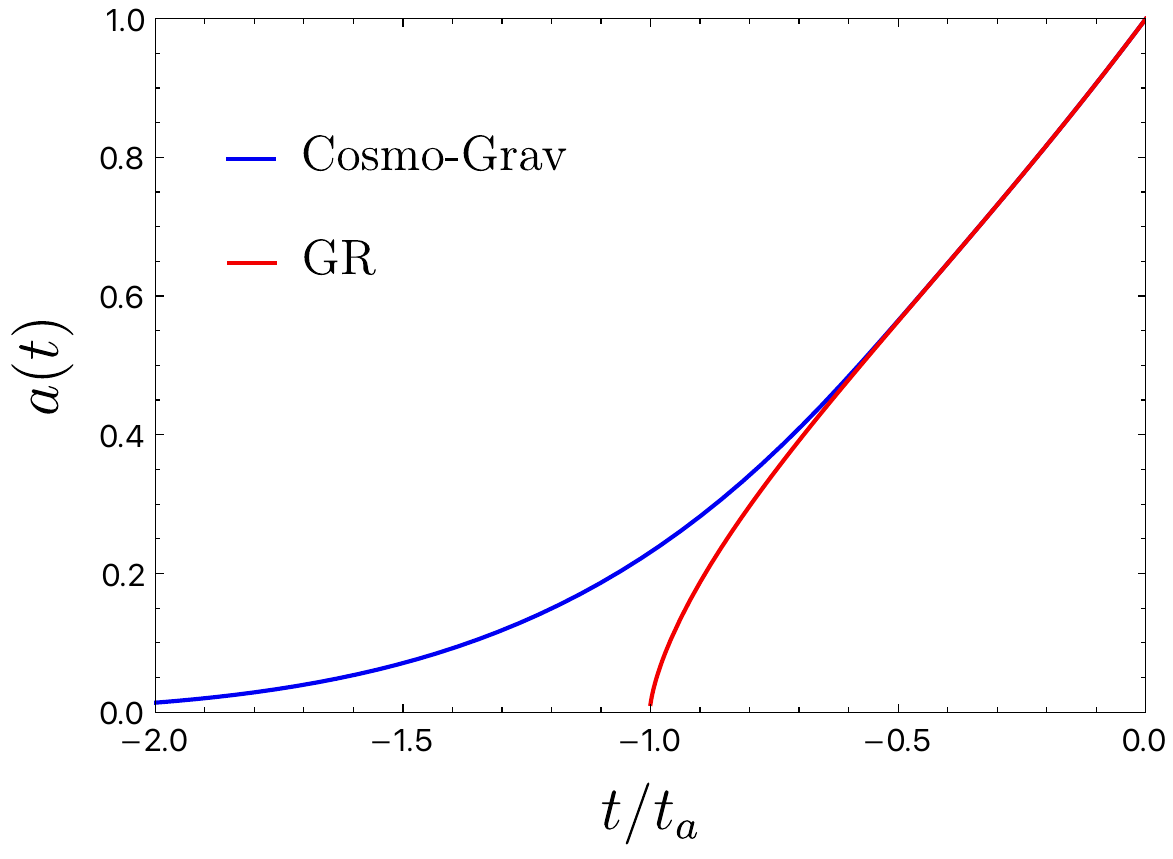}
\caption{Dynamical evolution of the scale factor in General Relativity (denoted in the plot as ``GR'') and in our Cosmological Gravities (denoted as ``Cosmo-Grav''), including up to sixth-order corrections in the curvature. We have used the values for the couplings $\mu_n$ obtained in \eqref{eq:couplings}. Additionally, $t_a$ stands for the age of the Universe and $t=0$ stands for present day.}
\label{fig:1}
\end{figure}

Two comments are in order. First, we note that the inclusion of higher-curvature corrections smooths the GR singularity. Nevertheless, we have observed that it appears around $t \simeq -3.0 t_a$, where $t_a$ stands for the age of the Universe (according to GR). This was to be expected, since the full singularity resolution should only occur after the inclusion of the complete tower of higher-curvature terms\footnote{These aspects are particularly appealing for the problem of the Big Bang initial conditions.} (here we truncated the tower at sixth order in the curvature) \cite{Arciniega:2018tnn}. Second, we have checked that our model featuring higher-curvature terms naturally accounts for an inflationary era\footnote{Taking the beginning of inflation to be the moment in which the radiation density equals the Planck density, and the end of inflation to be the moment in which the the acceleration starts to decrease, we found that inflation occurred with a number of $\simeq$ 70 e-folds.} followed by a radiation-dominated epoch, a matter-dominated era and the current moment of late-time acceleration. This phenomenon of geometric inflation was already reported in \cite{Arciniega:2018tnn,Cisterna:2018tgx} and provides an intriguing explanation for inflation without the need of additional matter fields besides the metric. The novelty of our work lies in the fact that the scale factor presented in Figure \ref{fig:1} was obtained after choosing the couplings $\mu_n$ of our Cosmological Gravities to match with observations and then using such values to solve for the generalized Friedmann equations. Such a novel approach was not possible in \cite{Cisterna:2018tgx}, where no Cosmological Gravities beyond fourth order in the curvature were employed, nor in \cite{Arciniega:2018tnn}, where emphasis was placed on the replacement of the radiation-dominated early Universe by an epoch of exponential growth of the scale factor.


\section{Cosmological Generalized Quasitopological Gravities}\label{sec:CGQG}


Once we have found an explicit representative of the unique non-trivial inequivalent Cosmological Gravity existing at each curvature order in all dimensions $D\geq 3$, our next goal will be to find other representatives which satisfy additional properties when evaluated on different gravitational backgrounds.

We will be interested in looking for Cosmological Gravities which admit, furthermore, a special type of non-hairy generalizations of the Schwarzschild-Tangherlini black-hole solution. In particular, we will concentrate on those Cosmological Gravities for which the aforementioned static and spherically symmetric black holes are fully characterized by a single function whose equation of motion can be analytically integrated into a second-order one. Such static and spherically symmetric configurations may be written as follows\footnote{Observe that this is not the most general static and spherically symmetric ansatz one may think of. Indeed, it is only possible to perform coordinate transformations so as to write such a metric as $\diff s^2=-N(r)^2 f(r) \diff t^2+\frac{1}{f(r)} \diff r^2+r^2 \diff \Omega^2_{D-2}$. Equation \eqref{eq:sss} corresponds to the special case $N(r)=1$, as in the Schwarzschild-Tangherlini solution.}:
\begin{equation}
\diff s_f^2=-f(r) \diff t^2+\frac{1}{f(r)} \diff r^2+r^2 \diff \Omega^2_{D-2}\,,
\label{eq:sss}
\end{equation}
where $\diff \Omega^2_{D-2}$ is the metric of the round $D-2$ sphere. Higher-curvature gravities (non-necessarily cosmological) enjoying these features have been extensively explored in the literature under the name of Generalized Quasitopological Gravities (GQTGs) \cite{Oliva:2010eb,Myers:2010ru,Dehghani:2011vu,Bueno:2016xff,Hennigar:2017ego,Bueno:2019ltp,Bueno:2019ycr,Bueno:2022res,Aguilar-Gutierrez:2023kfn}. The space of GQTGs may be further subdivided into the class of Quasitopological Gravities, for which the equation of motion for $f(r)$ can be integrated into a purely algebraic one, and the class of proper GQTGs, for which such an integrated equation is strictly of second order in derivatives of $f(r)$. Additional relevant properties satisfied by GQTGs are that they allow for the analytic computation of black-hole thermodynamics, possess second-order linearized equations of motion on top of maximally symmetric backgrounds and form a perturbative basis for the space of effective theories of gravity \cite{Bueno:2016xff,Hennigar:2016gkm,Bueno:2016lrh,Hennigar:2017ego,Bueno:2017qce,Cisterna:2017umf,Bueno:2017sui,Bueno:2019ltp,Bueno:2019ycr,Cano:2019ozf,Bueno:2022res}.


\begin{defi}
A theory $\cL( g^{ab},R_{a b c d})$ is said to be a Cosmological GQTG if and only if it is both a Cosmological Gravity and a GQTG. 
\end{defi}

Specifically, a Cosmological GQTG possesses equations of motion of second order in derivatives for the function $f(r)$ in static and spherically symmetric configurations \eqref{eq:sss} and for the scale factor $a(t)$ in FLRW backgrounds  \eqref{eq:flrwansf1}. If the theory is of the quasitopological class, we will coin the name Cosmological Quasitopological Gravity.

Apart from the somewhat trivial case of Einstein gravity (and Lovelock gravities for generic spacetime dimension), explicit examples of cosmological GQTGs have been explicitly found in $D=4$. 
\begin{example}
Consider again the four-dimensional cubic theory proposed in Reference \cite{Arciniega:2018fxj}, which we rewrite here for the benefit of the reader:
\begin{align}
\mathcal{L}_{\mathrm{AEJ}}&=R-2 \Lambda+ \alpha \ell^4 (\mathcal{P}-8 \hat{\mathcal{P}})\,,
\end{align}
where the combinations of cubic invariants $\mathcal{P}$ and $\hat{\mathcal{P}}$ were presented in Equations \eqref{eq:paejdef} and \eqref{eq:phataejdef}. We already explained in Example \ref{ep:aej} that this theory is a Cosmological Gravity. However, by direct computation one observes that $\hat{\mathcal{P}}$ is identically zero when evaluated on generic static and spherically symmetric backgrounds, and one may also notice that $\mathcal{P}$ actually coincides with the celebrated Einsteinian Cubic gravity \cite{Bueno:2016xff}, the first proper GQTG ever identified. Consequently, we conclude that $\mathcal{L}_{\mathrm{AEJ}}$ is an instance of Cosmological GQTG.
\end{example}

In the previous example, we implicitly used the notion of trivial GQTGs. They are defined as combinations of curvature invariants which do not contribute to the subsequent equations of motion of the function $f(r)$ in Equation \eqref{eq:sss}. They will be crucial in the following, since they will allow us to convert non-cosmological GQTGs into Cosmological GQTGs, as we will momentarily see. 

Having said this, now we proceed to the problem of finding an instance of Cosmological GQTG at each curvature order and for all dimensions\footnote{In $D=3$, the only non-trivial GQTG is Einstein gravity \cite{Bueno:2022lhf}, so we opt not to consider the three-dimensional case in this manuscript.} $D \geq 4$. Previous literature in the topic focused on the four-dimensional case and identified Cosmological GQTGs up to eighth order in the curvature. Nevertheless, no closed expression valid for arbitrary orders in the curvature is known. Moreover, in dimensions $D \geq 5$, the problem of finding Cosmological GQTGs seems not to have been previously explored (with the celebrated exception of cubic-curvature theories, for which it was noted the existence of a Quasitopological Gravity of the cosmological type \cite{Bueno:2022log}). First, we devote ourselves to the identification of Cosmological GQTGs in dimensions $D \geq 5$ at all curvature orders and, afterwards, we concentrate on the four-dimensional case, finding instances of Cosmological GQTGs at all orders. 


\subsection{Cosmological GQTGs in $D \geq 5$}

Let us begin by presenting some notation, extracted from Reference \cite{Moreno:2023rfl}. We define the following curvature invariants:
\begin{align}
\label{eq:dd1}
\mathsf{W}_2 &=\frac{4}{(D-2)^2(D-1)(D-3)} W_{a b c d} W^{a b c d}\,, \\
\label{eq:dd2}
\mathsf{Z}_2 & =\frac{2}{D(D-2)} Z^a_b Z_a^b\, ,\\\
\label{eq:dd3}
\mathsf{W}_3 &= \frac{8}{(D-3)(D-2)(2(2-(D-3)^2)+(D-2)^2(D-3)^2)} W_{ab}{}^{cd}W_{cd}{}^{ef}W_{ef}{}^{ab}\,,\\
\label{eq:dd4}
\mathsf{Y}_3 & = \frac{8}{D^2(D-2)(D-3)}  W_{ac}{}^{bd} Z^a_b Z^c_d\, , \\
\label{eq:dd5}
\mathsf{X}_3 & = -\frac{8}{(D-1)^2(D-2)(D-3)(D-4)}  W_{a c d e}W^{bcde} Z^a_b\,,\\
\label{eq:dd6}
\mathsf{Z}_3 & = -\frac{4}{D(D-2)(D-4)} Z^a_b Z^b_cZ_a^c\, , \\
\label{eq:dd7}
\mathsf{Y}_4 & = - \frac{16}{D^2(D-2)(D-3)(D-4)}  Z^{a}_b Z_{a c} W^{bdce} Z_{d e}  \, ,\\
\label{eq:dd8}
\mathsf{X}_4 & = -\frac{32}{D(D-1)^2(D-2)(D-3)^2(D-4)} Z^{ab} W_{acbd} W^{c efg} W^d{}_{efg}\, .
\end{align}
Remember that $D \geq 5$. Also, we introduce:
\begin{align}
\I^{(1)}_l&=\mathsf{W}_2^{\frac{l-\pi_l}{2}} \left ( (1-\pi_l) \mathsf{W}_2 +\pi_l \sW_3\right)\, , \\
\I^{(2)}_l&=\mathsf{Z}_2^{\frac{l-\pi_l}{2}} \left (  (1-\pi_l) \mathsf{Z}_2 +  \pi_l \mathsf{Z}_3 \right)\,, \\
\I^{(3)}_l &=\mathsf{W}_2^{\frac{l-\pi_l}{2}} \left ( (1-\pi_l) \mathsf{X}_3 +  \pi_l \mathsf{X}_4\right)\, , \\
\I^{(4)}_l&=\mathsf{Z}_2^{\frac{l-\pi_l}{2}} \left (  (1-\pi_l) \mathsf{Y}_3 +  \pi_l \mathsf{Y}_4 \right)\,,
\end{align}
where
\begin{equation}
\pi_l= \left \lbrace \begin{matrix}
0\,  \quad l \,\, \mathrm{even}\\ 
1\,  \quad l \,\, \mathrm{odd}
\end{matrix} \right.\,.
\end{equation}
Now consider the following order-$n$ curvature invariants:
\begin{align}
\mathcal{Z}_{(1)}&=R\,, \\
\mathcal{Z}_{(2)}&=R^2+\frac{(D-1)D W_{abcd}W^{abcd}}{(D-3)(D-2)}-\frac{4D(D-1) Z_{ab}Z^{ab}}{(D-2)^2}\,, \\
\mathcal{Z}_{(n)}&=R^n+ \sum_{l=0}^{n-2}  R^{n-l-2}\left (\gamma_{n,-2,l}\mathcal{I}^{(1)}_l+\gamma_{n,l,-2}\mathcal{I}^{(2)}_l\right)+\sum_{l=0}^{n-3}  R^{n-l-3}\left (\gamma_{n,-1,l} \mathcal{I}^{(3)}_l+\gamma_{n,l,-1}\mathcal{I}^{(4)}_l\right)\nonumber \\ &+ \sum_{l=0}^{n-4}\sum_{p=0}^{n-l-4} \gamma_{n,l,p} R^{n-l-p-4}\mathcal{I}^{(1)}_p \mathcal{I}^{(2)}_l\,, \quad n \geq 3 \,, \label{eq:noncosmoQTG}
\end{align}
where $\gamma_{n,l,p}$ is only non-zero for $l,p \geq -2$ and $l+p +4 \leq n$, in which case
\begin{align}
\label{eq:gammaconst}
\gamma_{n,l,p}&=\frac{n!(D(D(l-2)+4)(l+1)+4(D-1)(D l+1)(p+2)+4(D-1)^2(p+2)^2)}{2^{2-l+p} (D^2-D)^{-p-l-3} (D-2)^{l+2} (l+2)!(p+2)!(n-l-p-4)! }\,.
\end{align}
It was shown in Reference \cite{Moreno:2023rfl} that $\mathcal{Z}_{(n)}$ defines a Quasitopological gravity for all $D \geq 5$ and curvature order $n \geq 1$. Now we are going to try to build Cosmological GQTGs of order $n$ in the curvature by adding trivial GQTGs (i.e., combinations of curvature invariants which do not contribute to the equation of $f(r)$ in Equation \eqref{eq:sss}) to $\mathcal{Z}_{(n)}$. It is important to note that this is a quite adventurous proposal. Indeed, for $D \geq 5$  there exist $n-1$ GQTGs at each curvature order $n$ providing linearly independent equations of motion for the function $f(r)$, only one of the quasitopological class \cite{Bueno:2022res,Moreno:2023rfl}. Therefore, it is not clear at all that starting solely from the Quasitopological Gravity $\mathcal{Z}_{(n)}$ and deforming it through the addition of trivial terms will result in a Cosmological Quasitopological Gravity. However, such an idea turns out to be successful, as we proceed to show.

For that, let us consider the following trivial GQTGs:
\begin{align}
\nonumber
\mathcal{T}_{(n+2)}^{\mathcal{Z}}= \frac{2  (\mathcal{K}_2)^{\frac{n-2-\pi_n}{2}}}{D^3(D^2-4D+3)}& \left ((D-2)\left (Z_{a}^bZ_{b}^c Z_c^d Z_d^a-\frac{D^2-6D+12}{2D(D-2)} (Z_a^b Z_b^a)^2 \right) (1-\pi_n) \right. \\ \label{eq:trivialgqtgeq} & \left.-\left ( Z_{a}^bZ_{b}^c Z_c^d Z_d^e Z_e^a-\frac{D(D-4)+8}{2D(D-2)} Z_a^b Z_b^a Z_c^d Z_d^e Z_e^c \right )\pi_n  \right )\,, \quad n \geq 0\,.
\end{align}
They satisfy $\left. \mathcal{T}_{(n+2)}^{\mathcal{Z}}\right \vert_{\diff s^2_f}=0$, so they are trivial GQTGs. However, $\mathcal{T}_{(n+2)}^{\mathcal{Z}}$ does not vanish when evaluated on FLRW backgrounds \eqref{eq:flrwansf1}, so it does contribute to the equations of motion of cosmological backgrounds. In particular, we have that
\begin{equation}
\left. \mathcal{T}_{(n+2)}^{\mathcal{Z}} \right \vert_{\diff s^2_{a,F}}= \Gamma^{n+2}\,,
\end{equation}
where $\Gamma$ was defined back at Equation \eqref{eq:gammadef}. Having said this, now observe that the Quasitopological Gravities presented in Equation \eqref{eq:noncosmoQTG} have a quite simple form when evaluated on the FLRW background \eqref{eq:flrwans}:
\begin{equation}
\left. \mathcal{Z}_{(n)} \right \vert_{\diff s^2_{a,F}}=\Sigma^n+ \sum_{l=0}^{n-2} \kappa_{n,l} \Sigma^{n-l-2} \Gamma^{l+2}\,,
\end{equation}
where 
\begin{equation}
\kappa_{n,l}= \binom{n}{l+2}\frac{D^{l+2} (l+1)((l-2)D+4))}{4(D-1)} \left (\frac{2(D-1)}{D-2} \right)^{\frac{3l+6-\pi_l}{2}} \left (1+ \frac{D}{D-4} \pi_l \right)\,.
\end{equation}
Interestingly enough, it turns out that $\mathcal{Z}_{(n)}$ for $n \leq 3$ is already a Cosmological Gravity for every $D \geq 5$, without adding further curvature invariants. To find Cosmological GQTGs (of the quasitopological class) at all curvature orders $n \geq 4$, let us consider:
\begin{equation}
\mathsf{Q}_{(n)}= \mathcal{Z}_{(n)} + \sum_{l=0}^{n-2} \eta_{n,l}  R^{n-l-2}\mathcal{T}_{(l+2)}^{\mathcal{Z}} \,, \quad  n \geq 4
\label{eq:teoriacuasiaunno}
\end{equation}
where $\eta_{n,l}$ are certain constants. The evaluation on top of the FLRW background \eqref{eq:flrwans} reads:
\begin{equation}
\mathsf{Q}_{(n)} \vert_{\diff s^2_{a,F}}=\left. \mathcal{Z}_{(n)} \right \vert_{\diff s^2_{a,F}}+ \sum_{l=0}^{n-2} \eta_{n,l} \Sigma^{n-l-4} \Gamma^{l+4}\,.
\label{eq:qquasicond}
\end{equation}
Now, according to Proposition \ref{prop:cosmocond}, $\mathsf{Q}_{(n)}$ will be a Cosmological Gravity if and only if:
\begin{equation}
\frac{\partial^2 \mathsf{Q}_{(n)} \vert_{\diff s^2_{a,F}} }{ \partial \ddot{a}^2}=0\,.
\label{eq:cosmocondquasi}
\end{equation}
After some computations, it can be seen that this requirement uniquely fixes the coefficients $\eta_{n,l}$ to be:
\begin{align}
\nonumber  \eta_{n,l}=-\frac{2^{l+1} (l+1) (D-1)^{l+2} D^{l+2}}{(D-2)^{l+3}} & \label{eq:etaquasi} \binom{n}{l+2} \Bigg (2(D-2)\\  & \hspace{-3cm}+\left. \frac{(D(l-2)+4)(D(l+1-2 \lfloor l/2 \rfloor)-4)}{D-4} \left ( \frac{2D-2}{D-2} \right)^{\lfloor l/2 \rfloor} \right) \,,
\end{align}
where $\lfloor x \rfloor$ denotes the largest integer less or equal than $x$  and $n \geq 4$. With this choice, $\left. \mathsf{Q}_{(n)} \right \vert_{\diff s^2_{a,F}}$ will automatically satisfy condition \eqref{eq:cosmocondquasi}. Since $\mathcal{T}_{(l+2)}^{\mathcal{Z}}$ is a trivial GQTG, we arrive to the following theorem.

\begin{theorem}
Let us consider the $n$-th order higher-curvature gravities:
\begin{align}
\mathsf{Q}_{(n)}&=\mathcal{Z}_{(n)}\,,  \quad 1 \leq n \leq 3 \\
\label{eq:quasicosmo}
\mathsf{Q}_{(n)}&=\mathcal{Z}_{(n)} + \sum_{l=0}^{n-2} \eta_{n,l}  R^{n-l-2}\mathcal{T}_{(l+2)}^{\mathcal{Z}}\,, \quad n \geq 4\,,
\end{align}
where the constants $\eta_{n,l}$ are given by Equation \eqref{eq:etaquasi} and $\mathcal{T}_{(l+4)}^{\mathcal{Z}}$ denotes the trivial GQTGs presented in Equation \eqref{eq:trivialgqtgeq}. Then,  $\mathsf{Q}_{(n)}$ is a Cosmological Quasitopological Gravity for all $D \geq 5$ and $n \geq 1$.
\end{theorem}
To the best of our knowledge, these represent the first instances ever of Cosmological Quasitopological Gravities, which will also satisfy a holographic $c$-theorem by construction. Given that Quasitopological Gravities form a very particular subset of the whole space of GQTGs in $D \geq 5$, it is remarkable that adding some trivial GQTGs to the theories $\mathcal{Z}_{(n)}$ given in Equation \eqref{eq:noncosmoQTG} produces Cosmological GQTGs at all curvature orders $n$ and dimensions $D \geq 5$. Indeed, one is tempted to talk about a \emph{Quasitopological miracle}. 

For low orders, the Cosmological GQTGs we have found in generic $D \geq 5$ read as follows:
\begin{align}
\mathsf{Q}_{(1)}&=R\,, \\
\mathsf{Q}_{(2)}&=R^2+ \frac{D(D-2)(D-1)^2}{4} W_{abcd} W^{abcd}- \frac{2D^2(D-1)}{D-2} Z_{ab} Z^{ab}\,, \\
\mathsf{Q}_{(3)}&=R^3+\frac{2 (D-1)^2 D^2 (2 D-3) W\indices{_a_b^c^d}W\indices{_c_d^e^f}W\indices{_e_f^a^b}}{(D-3) (D-2) (D ((D-9)
   D+26)-22)}+\frac{24 (D-1)^2 D^2 Z^a_b Z^c_d W\indices{_a_c^b^d}}{(D-3) (D-2)^3}\notag\\
   &+\frac{16 (D-1)^2
   D^2 Z^a_b Z^b_cZ_a^c}{(D-2)^4}-\frac{24 (D-1)^2 D^2 Z^a_b W_{a c d e}W^{bcde}}{(D-4) (D-3)
   (D-2)^2}+\frac{3 (D-1) D R W_{a b c d} W^{a b c d}}{(D-3) (D-2)}\notag\\
   &-\frac{12 (D-1) D R
   Z^a_b Z_a^b}{(D-2)^2}\, ,\\ \nonumber
   \mathsf{Q}_{(4)}&=R^4+\frac{3 (D-1)^2 D^3 (3 D-4) \left(W_{a b c d} W^{a b c d}\right)^2}{(D-3)^2
   (D-2)^4}- \frac{384 (D-1)^3 D^3 Z^{a}_b Z_{a c} Z_{d e} W^{bdce}}{(D-4)
   (D-3) (D-2)^4}\\\nonumber &+\frac{48 (D-1)^3 D^3
  \left(Z^a_b Z_a^b\right)^2}{(D-3) (D-2)^5}+\frac{8 (D-1)^2 D^2 (2
   D-3) R  W\indices{_a_b^c^d}W\indices{_c_d^e^f}W\indices{_e_f^a^b}}{(D-3) (D-2) (D ((D-9)
   D+26)-22)}\\\nonumber & -\frac{96 (D-1)^3 D^3
  Z^a_bZ^b_cZ^c_dZ^d_a}{(D-3) (D-2)^5} +\frac{96 (D-1)^2 D^2 R Z^a_b Z^c_d W\indices{_a_c^b^d}}{(D-3)
   (D-2)^3}+\frac{64 (D-1)^2 D^2 R
   Z^a_b Z^b_cZ_a^c}{(D-2)^4}\\\nonumber & -\frac{96 (D-1)^2 D^2 R
   Z^a_b W_{a c d e}W^{bcde}}{(D-4) (D-3) (D-2)^2}+\frac{24 (D-1)^2
   D^2 (D (7 D-10)+4) W_{a b c d} W^{a b c d}Z_{ef} Z^{ef}}{(D-3)
   (D-2)^5}\\\nonumber &-\frac{192 (D-1)^3 D^2 Z^{ab} W_{acbd} W^{c efg} W^d{}_{efg}}{(D-4)
   \left(D^2-5 D+6\right)^2}+\frac{6 (D-1) D R^2
 W_{a b c d} W^{a b c d}}{(D-3) (D-2)}\\&-\frac{24 (D-1) D R^2
   Z_{ab} Z^{ab}}{(D-2)^2}\,.
\end{align}
Higher-order Cosmological Quasitopological Gravities in generic dimensions $D \geq 5$ may be found in Appendix \ref{App:A}. 

In another vein, it is important to note that it is possible to produce Cosmological GQTGs by adding trivial GQTGs to combinations of proper GQTGs. To illustrate this feature, let us work momentarily in $D=5$ and consider the following proper quartic GQTG:
\begin{align}
\nonumber
\mathcal{S}_{(4)}^{(D=5)}&=R^4+360 R^2 \sW_2-126000 \sW_2+25600 R \sX_3+11200 R \sW_3-224000 \sX_4+\frac{40000}{3} R \sY_3\\&+ \frac{200000}{3} \sY_4-400 R^2 \sZ_2-56000 \sW_2 \sZ_2+\frac{640000}{27} \sZ_2^2-\frac{32000}{27} R \sZ_3\,.
\end{align}
If we now construct the theory:
\begin{equation}
\mathsf{Q}'_{(4)}=\mathcal{S}_{(4)}^{(D=5)}-\frac{128000}{81}\left (Z_{a}^bZ_{b}^c Z_c^d Z_d^a-\frac{7}{30} (Z_a^b Z_b^a)^2\right) \,,
\label{eq:propgqg5}
\end{equation}
one can readily check that this defines a Cosmological GQTG in $D=5$ which is not of the quasitopological type. Therefore, proper GQTGs may also be used to obtain Cosmological GQTGs. All these results show that the complete characterization of the space of Cosmological GQTGs in $D \geq 5$ is a quite rich and interesting problem to be explored in the future.
\subsection{Four-dimensional Cosmological GQTGs}


Let us now explore the problem of finding Cosmological GQTGs in $D=4$. The existence of such theories has been established up to eighth order in the curvature \cite{Arciniega:2018fxj,Cisterna:2018tgx,Arciniega:2018tnn}. Our goal is extend this result to all curvature orders $n$ by finding an example of a four-dimensional Cosmological GQTG for every $n$. 

To this aim, let us start by presenting an instance of a (proper) GQTG in $D=4$ at any curvature order $n \geq 1$ \cite{Moreno:2023rfl}:
\begin{align}
\mathcal{S}_{(1)}&=R\,, \\
\mathcal{S}_{(2)}&=R^2+6 W_{abcd} W^{abcd}-12 Z_{ab}Z^{ab}\,, \\
\nonumber
\mathcal{S}_{(n)}&=R^n-6n(n-1) R^{n-2}Z_2+18n(n-1)(n-2) R^{n-3}W_{abcd} Z^{ac} Z^{bd} \\&+ \sum_{l=0}^{n-2} \lambda^{(1)}_{l} R^{n-l-4} \mathcal{W}_l \left ( R^2+\lambda_	l^{(2)} Z_{2} \right)\,, \quad n \geq 3\,,
\end{align}
where $\mathcal{S}_{(2)}$ coincides with the Gauss-Bonnet invariant and
\begin{align}
Z_2&=Z_{ab}Z^{ab}\, , \quad Z_3=Z_{ab}Z^{b}{}_c Z^{ca}\,, \\
\mathcal{W}_l&=\left (W_{mnop} W^{mnop} \right)^{\frac{l-\pi_l}{2}}((1-\pi_l) W_{abcd} W^{abcd}+2 \pi_l  W_{abcd} W^{cdef}W_{ef}{}^{ab})\,, \\
\lambda^{(1)}_{l}&=\frac{(-1)^l 3^{l/2+\pi_l/2+1}(l+1)(3l+4)n!}{2(l+2)!(n-l-2)!}\,, \quad \lambda^{(2)}_{l}=-12\frac{(n-l-2)(n-l-3)}{(l+1)(3l+4)}\,.
\end{align}

The strategy to find instances of four-dimensional Cosmological GQTGs at arbitrary curvature order is the same as in previous section: add appropriate trivial GQTGs which do not vanish on FLRW background and allow us to satisfy the cosmological condition, in conjunction with the GQTG definition. To this aim, let us consider the following combination of curvature invariants:
\begin{align}
\label{eq:trivialgqtg43}
\mathcal{T}^{\mathcal{S}}_{(3)}&=144 Z_3\,, \quad \\
\mathcal{T}^{\mathcal{S}}_{(n)}&=\frac{8\omega_{n,3}}{3}  R^{n-3} Z_3 + \sum_{l=4}^n \omega_{n,l} R^{n-l} \left (Z_2 \right)^{\frac{l-4-\pi_l}{2}} \left ((1-\pi_l)\left (Z_2^2 -4 Z_4\right )  +\frac{8 \pi_l}{3}  Z_2 Z_3 \right)\,, \quad n \geq 4\,,
\label{eq:trivialgqtg4}
\end{align}
where  
\begin{equation}
Z_4=Z_{ab}Z^{b}{}_c Z^{c}{}_d Z^{da} \,, \quad \omega_{n,l}= 2^{l-2-\pi_l} 3^{(2+l+\pi_l)/2}   \binom{n}{l} (l-1) \,.
\end{equation}
By direct computation one checks that $\mathcal{T}^{\mathcal{S}}_n$ is a trivial GQTG for every $n \geq 3$. Also, one may readily see that $\mathcal{S}_{(n)}+ \mathcal{T}^{\mathcal{S}}_n$ satisfies Equation \eqref{eq:cosmoteo} for every $n$, so it also belongs the cosmological class. Thus we prove the following theorem.
\begin{theorem}
\label{teo:cgqtgsallorders4}
For any curvature order $n$, consider the following higher-curvature gravity:
\begin{align}
\mathcal{Q}_{(1)}&=R\,, \\
\mathcal{Q}_{(2)}&=R^2+6 W_{abcd} W^{abcd}-12 Z_{ab}Z^{ab}\,, \\ \label{eq:cosmogqgteo}
\mathcal{Q}_{(n)}&=\mathcal{S}_{(n)}+ \mathcal{T}^{\mathcal{S}}_{(n)}\,, \quad n \geq 3
\end{align}
with $\mathcal{T}^{\mathcal{S}}_{(n)}$ given by Equations  \eqref{eq:trivialgqtg43} and \eqref{eq:trivialgqtg4}. It defines an example of Cosmological GQTG at all curvature orders.
\end{theorem}

Consequently, we have managed to extend the existence of Cosmological GQTGs for all curvature orders. Equivalently, we may claim to have found instances of GQTGs of arbitrary order which furthermore satisfy a holographic $c$-theorem. 

Expression \eqref{eq:cosmogqgteo} provides the following Lagrangians at the lowest curvature orders:
\begin{align}
\mathcal{Q}_{(3)}&=R^3+18 R W_{a b c d} W^{a b c d}-36 R Z^a_b Z_a^b-126 W\indices{_a_b^c^d}W\indices{_c_d^e^f}W\indices{_e_f^a^b}+108 Z^a_b Z^c_d W\indices{_a_c^b^d}\notag\\
&+144 Z^a_b Z^b_cZ_a^c\, ,\\
\mathcal{Q}_{(4)}&=R^4+36 R^2 W_{a b c d} W^{a b c d}-72 R^2 Z^a_b Z_a^b-504 R W\indices{_a_b^c^d}W\indices{_c_d^e^f}W\indices{_e_f^a^b}+432 R Z^a_b Z^c_d W\indices{_a_c^b^d}\notag\\
&+576 R Z^a_b Z^b_cZ_a^c+135 \left(W_{a b c d} W^{a b c d}\right)^2-216 W_{a b c d} W^{a b c d}Z^e_f Z_e^f+324 \left(Z^a_b Z_a^b\right)^2\notag\\
&-1296Z^a_bZ^b_cZ^c_dZ^d_a\, ,\\
\mathcal{Q}_{(5)}&=R^5+60 R^3 W_{a b c d} W^{a b c d}-120 R^3 Z^a_b Z_a^b-1260 R^2 W\indices{_a_b^c^d}W\indices{_c_d^e^f}W\indices{_e_f^a^b}\notag\\
&+1080 R^2 Z^a_b Z^c_d W\indices{_a_c^b^d}+1440 R^2 Z^a_b Z^b_cZ_a^c+675 R \left(W_{a b c d} W^{a b c d}\right)^2+1620 R \left(Z^a_b Z_a^b\right)^2\notag\\
&-1080 R W_{a b c d} W^{a b c d}Z^e_f Z_e^f-6480 R
   Z^a_bZ^b_cZ^c_dZ^d_a-1404 W_{a b c d} W^{a b c d}W\indices{_e_f^g^h}W\indices{_g_h^i^j}W\indices{_i_j^e^f}\notag\\
   &+2160 W\indices{_a_b^c^d}W\indices{_c_d^e^f}W\indices{_e_f^a^b}Z^g_hZ^h_g+3456 Z^a_b Z_a^bZ^c_d Z^d_eZ_c^e\, ,
\end{align}
In the appendix we present additional expressions for $\mathcal{Q}_{(n)}$ up to $n=9$. 

Let us now relate the cubic and quartic theories we have obtained to other theories that have already appeared in the literature. On the one hand, the cubic Lagrangian $\mathcal{Q}_{(3)}$ turns out to be related to the theory presented in Equation \eqref{eq:AEJ} through the expression
\begin{equation}\label{eq:rel3}
\frac{1}{9}\mathcal{Q}_{(3)}=-2(\mathcal{P}-8\hat{\mathcal{P}})\,.
\end{equation}
As a matter of fact, up to constant factors, $\mathcal{Q}_{(3)}$ represents the unique combination of cubic invariants which forms a Cosmological GQTG. 

On the other hand, Cosmological GQTGs of quartic order have been explored in Reference \cite{Cano:2020oaa}. In their writing, they list three examples of quartic Cosmological GQTGs: $\mathcal{R}_{(4)}^A$, $\mathcal{R}_{(4)}^B$ and $\mathcal{R}_{(4)}^C$ (See Equations (5-7) in Reference \cite{Cano:2020oaa}). The relation between their theories and our $\mathcal{Q}_{(4)}$ is given by
\begin{align}
\label{eq:pabloynuestra}
\mathcal{Q}_{(4)}&=108\left(-32\mathcal{R}_{(4)}^A+16\mathcal{R}_{(4)}^B-7\mathcal{R}_{(4)}^C\right)-432\mathcal{T}_{(4)}'\, , \\
\label{eq:trivialnuestra}
\mathcal{T}_{(4)}'&=Z^{af}Z_f^bW_{acde}W_{b}{}^{cde}-3Z^{ab}Z^{cd}W_{ac}{}^{ef}W_{efbd}\,,
\end{align}
$\mathcal{T}_{(4)}'$ being a trivial GQTG of quartic order.




\section{Cosmological Perturbations in Cosmological Gravities}
\label{sec:cosmopert}
In this final section of the manuscript we devote ourselves to the study of cosmological perturbations in the context of Cosmological Gravities. As it is well known, cosmological perturbations may be divided into scalar, vector and tensor perturbations. On the one hand, vector perturbations can be argued to be negligible in an expanding Universe --- such as ours ---, while tensor perturbations are associated to gravitational waves \cite{baumann2022cosmology}. On the other hand, scalar perturbations couple to both density and pressure perturbations and are responsible for the formation of structures in the Universe, so we decide to focus on them.

In the conformal-Newtonian gauge, scalar perturbations of FLRW backgrounds may be encompassed in two scalar functions $\Phi(\eta,x^i)$ and $\Psi(\eta,x^i)$ which modify the FLRW ansatz as follows:
\begin{equation}
\diff s^2=a^2(\eta)\left[ -(1+2\Phi) \diff \eta^2+(1-2\Psi) \gamma_{ij} \diff x^i \diff x^j \right]\,,
\label{eq:cosmopert}
\end{equation}
where $\eta$ stands for the conformal time\footnote{In absence of perturbations, a simple time reparametrization allows us to change among the conformal time $\eta$ and the usual comoving time $t$ we have been using along the document} and $\gamma_{ij} \diff x^i \diff x^j$ stands either for the metric of the $(D-1)$-dimensional round sphere, hyperbolic space or plane space. 

There is a crucial property satisfied by all Cosmological Gravities in all dimensions $D \geq 3$. Not only they possess second-order equations of motion for the scale factor $a(t)$ on FLRW backgrounds, but also the equations of motion for the scalar perturbations $\Phi$ and $\Psi$ are strictly of second-order in time derivatives. This feature was already observed to hold for cubic Cosmological Gravities in $D=4$ in Reference \cite{Cisterna:2018tgx} and for certain quartic Cosmological Gravities in Reference \cite{Cano:2020oaa}, and now we proceed to prove it in complete generality.
\begin{theorem}
\label{theo:cosmo2pert}
Let $\mathcal{L}(g^{ab},R_{abcd})$ be any Cosmological Gravity. Then, the linear equations of motion for cosmological scalar perturbations of the form \eqref{eq:cosmopert} are of second order in time derivatives.
\end{theorem}
\begin{proof}
The general gravitational equations of motion for any higher-curvature gravity  $\mathcal{L}(g^{ab},R_{abcd})$ take the form \cite{Padmanabhan:2013xyr}:
\begin{equation}
\mathcal{E}_{ab}=P_{acde} R_b{}^{cde}-\frac{1}{2} \mathcal{L} g_{ab}+2 \nabla^c \nabla^d P_{acbd}=0\,,
\label{eq:geneom}
\end{equation}
where $P^{abcd}=\frac{\partial \mathcal{L}}{\partial R_{abcd}}$. From here, it is clear that all derivatives of metric components of higher-degree than two arise from the piece $\nabla^c \nabla^d P_{acbd}$. 

Assume that $\mathcal{L}(g^{ab},R_{abcd})$ is a Cosmological Gravity and consider the perturbed FLRW metric \eqref{eq:cosmopert}. First, we observe that the linear equations of motion for $\Phi$ and $\Psi$ may not contain any term with four temporal derivatives of $\Phi$ or $\Psi$. Indeed, if $\Phi=\Phi(\eta)$ and $\Psi=\Psi(\eta)$ had only time dependence, relabel $\eta \rightarrow t$ and associate $a^2(\eta) (1+2\Phi(\eta))$ with $F(t)$ and $a^2(\eta) (1-2\Phi(\eta))$ with $a^2(t)$ in Equation \eqref{eq:flrwans}. Since $\mathcal{L}(g^{ab},R_{abcd})$ is a Cosmological Gravity, no time derivatives of order larger than two may appear in the subsequent equations of motion, and this guarantees the absence of terms with four temporal derivatives of $\Phi$ or $\Psi$. As a byproduct, we also note that no terms of the form $\partial_\eta^3  \Phi$ or $\partial_\eta^3  \Psi$ may appear.


Hence all potential terms with more than two temporal derivatives must be of the form $\partial_\eta^3 \partial_i \Phi$ or $\partial_\eta^3 \partial_i \Psi$. Such pieces may arise either from mixed derivatives $\partial_\eta \partial_i \Phi$ or $\partial_\eta \partial_i \Psi$ in the tensor $P^{abcd}$ or from second-order time derivatives $\partial_\eta^2 \Phi$ and  $\partial_\eta^2 \Psi$. In the latter case, the potentially higher-derivative terms would arise from $\nabla^i \nabla^\eta P_{a i b\eta}+ \nabla^\eta  \nabla^i P_{a\eta b i}$. By commuting partial derivatives with the introduction of appropriate contractions with Riemann curvature tensors, one observes that  all possible higher-derivative terms $\partial_\eta^3 \partial_i \Phi$ or $\partial_\eta^3 \partial_i \Psi$ could arise from $\nabla^\eta  \nabla^i P_{a\eta b i}$. Noting that, at leading order, those terms in $P_{a\eta b i}$ already  linear in the perturbations may be considered to be acted by covariant derivatives of the unperturbed FLRW space,  $\nabla^\eta  \nabla^i P_{a\eta b i}$ would generate terms  $\partial_\eta^3  \Phi$ or $\partial_\eta^3  \Psi$, which are forbidden according to the argumentation above. Therefore, all terms $\partial_\eta^3 \partial_i \Phi$ or $\partial_\eta^3 \partial_i \Psi$ may only arise from terms containing $\partial_\eta \partial_i \Phi$ or $\partial_\eta \partial_i \Psi$ in the tensor $P^{abcd}$ evaluated at linear order in the perturbations. 


Now, let $R^{(1)}$, $Z^{(1)}{}_{ab}$ and $W^{(1)}{}_{ab}{}^{cd}$ denote the Ricci scalar, traceless Ricci tensor and Weyl curvature tensor at first order in the perturbations. The following properties hold:
\begin{align}
\frac{\partial R^{(1)}}{\partial (\partial_\eta \partial_i \Phi) }&=0\,, \quad \frac{\partial Z^{(1)}{}_{ab}}{\partial (\partial_\eta \partial_i \Phi) }=0\,, \quad \frac{\partial W^{(1)}{}_{ab}{}^{cd}}{\partial (\partial_\eta \partial_i \Phi) }=0\,,  \\
\frac{\partial R^{(1)}}{\partial (\partial_\eta \partial_i \Psi) }&=0\,, \quad \frac{\partial Z^{(1)}{}_{ab}}{\partial (\partial_\eta \partial_i \Psi) }=\mathfrak{f}_1(\eta) \delta^\eta_{(a} \delta^i_{b)} \,, \quad \frac{\partial W^{(1)}{}_{ab}{}^{cd}}{\partial (\partial_\eta \partial_i \Psi) }=0\,,
\end{align}
where $\mathfrak{f}_1(\eta)$ is a certain function whose specific form is unimportant for the purposes of this proof. These properties ensure that we just need to care about potential terms with $\partial_\eta \partial_i \Psi$ and, furthermore, we know that they may only arise from terms containing the traceless Ricci tensor at linear order in the perturbations. 

Let us then consider all terms that may contain $\partial_\eta \partial_i \Psi$ at linear order in the perturbations of a generic tensor $P_{abcd}$ of any higher-curvature gravity. Note the following properties:
\begin{align}
\label{eq:cosmoscalp1}
\delta^\eta_{(a} \delta^i_{b)} (c_1 \delta_{\eta}^b \delta_\eta^c+c_2 g^{bc})&=\frac{c_1}{2}\delta^i_{a} \delta_\eta^c+\frac{c_2}{2} (\delta^\eta_{a} g^{ic}+g^{\eta c}  \delta^i_{a}) + \mathcal{O}(\Phi,\Psi)\,, \\
\label{eq:cosmoscalp2}
\delta^\eta_{(a} \delta^i_{b)} (c_1 \delta_{\eta}^b \delta_\eta^a+c_2 g^{ba})&=\mathcal{O}(\Phi,\Psi)\,,
\end{align}
for any $c_1,c_2 \in \mathbb{R}$. These aspects imply that any potential piece\footnote{When dealing with tensors which are already of first order in the perturbations, we may lower and rise indices with the zeroth-order FLRW metric.} $\mathcal{A}_{abcd}$ in $P_{abcd}$ containing $\partial_\eta \partial_i \Psi$ must have the following schematic form:
\begin{equation}
\mathcal{A}^{abcd}=\mathfrak{f}_2(\eta) \Sigma^m \Gamma^p \left ( \diff Z \right)_{\eta i}{}^{abcd} \partial_\eta \partial_i \Psi +\mathcal{O}(\Phi^2,\Psi^2, \Phi \Psi)\,, 
\end{equation}
where $\Gamma$ and $\Sigma$ were defined back in  Equations \eqref{eq:gammadef}, $\mathfrak{f}_2(\eta)$ is a certain function and \eqref{eq:sigmadef} and
\begin{equation}
\left ( \diff Z \right)_{ef}{}^{abcd}=\frac{\partial Z_{ef}}{\partial R_{abcd}}=\delta_{(e}^{[b} g^{a][c} \delta_{f)}^{d]}-\frac{1}{D} g_{ef} g^{c[a} g^{b]d}\,.
\end{equation}
Terms arising from $\frac{\partial R}{\partial R_{abcd}}$ and $\frac{\partial W_{ijkl}}{\partial R_{abcd}}$ may be seen not to contribute to pieces with $\partial_\eta \partial_i \Psi$ from the properties of the traceless Ricci tensor of the unperturbed FLRW spacetime and Equations \eqref{eq:cosmoscalp1} and \eqref{eq:cosmoscalp2}.

Now, when we consider $\nabla^c \nabla^d \mathcal{A}_{acbd}$, the pieces with $\partial_\eta^3 \partial_i \Psi$ must arise from the covariant derivatives with $c=\eta$ and $d=\eta$ acting on $\partial_\eta \partial_i \Psi$  necessarily. However, from the properties of the unperturbed FLRW metric, it turns out that:
\begin{equation}
\left ( \diff Z \right)_{\eta i}{}^{a \eta b \eta}=\mathcal{O}(\Phi,\Psi)\,.
\end{equation}
Consequently, no third time derivatives of $\Phi$ or $\Psi$ may appear in the subsequent equations for the linear perturbations of Cosmological Gravities and we conclude. 
\end{proof}

Therefore, Cosmological Gravities do also yield equations of motion of second order in time derivatives for scalar cosmological perturbations. However, this does not exclude the presence of higher-derivative terms in the linear equations of motion for the perturbations, such as pieces affected by two time derivatives and two spatial derivatives. After all, Cosmological higher-curvature Gravities do not possess equations of motion of second-order for generic backgrounds, so one should expect the presence of such higher-derivative terms in the linear equations of motion for the perturbations\footnote{As noted in \cite{Pookkillath:2020iqq,BeltranJimenez:2020lee}, theories featuring a reduction of order of the equations of motion on certain highly symmetric configurations are, still, typically affected by instabilities associated to the so-called strong coupling problem. These issues disappear if we regard our models as effective gravitational theories \cite{Cano:2020oaa,Bueno:2023jtc,Jimenez:2023esa}, which is the point of view emphasized in the Introduction. In particular, computations and physical discussions resulting from these higher-curvature gravities would be perfectly valid as long as we retain ourselves within the small-coupling regime.}. In fact, we are going to explicitly check this by computing the equations of motion for the scalar perturbations for our four-dimensional Cosmological GQTGs derived in Equation \eqref{eq:cosmogqgteo}. This shall be the goal of next subsection.  

Nonetheless, before closing this part, there are two additional aspects that deserve to be mentioned. First, note that Cosmological Gravities which are identically vanishing on top of FLRW backgrounds might contribute to the equations for scalar perturbations, so one should pay attention to such type of terms. Second, one may wonder whether Cosmological Gravities do also produce linear equations of motion of second order in time derivatives for \emph{all} cosmological perturbations, including tensor and vector modes. However, direct check of the (trivial) Cosmological Gravity $W_{abcd} W^{abcd}$ shows that the equations of motion for tensor perturbations in the usual transverse and traceless gauge are of fourth order in time derivatives, thus showing that not all Cosmological Gravities may provide equations of second order in time derivatives for all perturbation modes. 

\subsection{Cosmological Perturbations of four-dimensional Cosmological GQTGs}

Once proven that every Cosmological Gravity features equations of motion of second order in time derivatives for scalar cosmological perturbations, we will present such equations for the four-dimensional Cosmological GQTGs we have obtained in Theorem \ref{teo:cgqtgsallorders4}. For the sake of simplicity, we will restrict ourselves to the case of transverse planar spatial 3-sections, so that $\gamma_{ij}$ can be taken to be Euclidean metric in Cartesian coordinates.

We begin by computing the tensor $P_{abcd}$ associated to the $n$-th order Cosmological GQTG $\mathcal{Q}_{(n)}$ presented in Equation \eqref{eq:cosmogqgteo}. It is  given by:
\begin{align}
P\indices{^a^b_c_d}&=nR^{n-1}\delta_{\phantom{[}c}^{[a}\delta_{d\phantom{]}}^{b]}-6n(n-1)\left((n-2)R^{n-3}\delta_{\phantom{[}c}^{[a}\delta_{d\phantom{]}}^{b]}Z_2+2R^{n-2}Z_{[c}^{[a}\delta_{d]}^{b]}\right)\notag\\
&+18n(n-1)(n-2)R^{n-3}\left(Z_{\phantom{[}c}^{[a}Z_{d\phantom{]}}^{b]}+\delta_{[c}^{[a}Z^{b]e}Z_{d] e}-\frac{1}{6}\delta_{\phantom{[}c}^{[a}\delta_{d\phantom{]}}^{b]}Z_2+2\delta_{[c}^{[a}W\indices{^{b]}^f_{d]}_e}Z_{\vphantom{[}f}^{\vphantom{[}e}\right)\notag\\
&+24n(n-1)(n-2)\left((n-3)R^{n-4}Z_3\delta_{\phantom{[}c}^{[a}\delta_{d\phantom{]}}^{b]}+3R^{n-3}\left(\delta_{[c}^{[a}Z^{b]e}Z_{d] e}^{\phantom{]}}-\frac{1}{4}\delta_{\phantom{[}c}^{[a}\delta_{d\phantom{]}}^{b]}Z_2\right)\right)\notag\\
&+6n(n-1)R^{n-4}W\indices{^a^b_c_d}\left(R^{2}-3(n-2)(n-3) Z_{2}\right)+\frac{1}{3}\sum_{l=4}^{n}\omega_{n,l}R^{n-l-1}Z_2^{\frac{l-6-\pi_l}{2}}\times\notag\\
&\Big(\left((n-l)\delta_{\phantom{[}c}^{[a}\delta_{d\phantom{]}}^{b]}Z_2+(l-4-\pi_l)RZ_{[c}^{[a}\delta_{d]}^{b]}\right)\left((1-\pi_l)\left(3Z_2^2-12Z_4\right)+8\pi_lZ_2Z_3\right)\notag\\
& +R Z_2\left(12(1-\pi_l)\left(\delta_{\phantom{[}c}^{[a}\delta_{d\phantom{]}}^{b]}Z_3+Z_{[c}^{[a}\delta_{d]}^{b]}Z_2-4\delta_{[c}^{[a}Z^{b]e}_{\phantom{]}}Z_{d]}^{f\phantom{]}}Z^{\phantom{]}}_{ef}\right.\right)\notag\\
&+2\pi_l\left((12\delta_{[c}^{[a}Z^{b]e}Z_{d] e}^{\phantom{]}}-3\delta_{\phantom{[}c}^{[a}\delta_{d\phantom{]}}^{b]}Z_2)Z_2+8Z_{[c}^{[a}\delta_{d]}^{b]}Z_3\right)\Big)\Big)+ \mathcal{O}(\Phi^2,\Psi^2, \Phi \Psi)\,,
\end{align}
where $\mathcal{O}(\Phi^2,\Psi^2, \Phi \Psi)$ stands for curvature invariants that become second-order in the perturbations if evaluated on \eqref{eq:cosmopert}. Notice that in this expression all terms beyond linear order in the Weyl tensor are $\mathcal{O}(\Phi^2,\Psi^2, \Phi \Psi)$. When evaluated on Equation  \eqref{eq:cosmopert}, one may explicitly check (by direct computation) that $\nabla^c \nabla^d P_{acbd}$ carries at most second-order temporal derivatives in the scalar degrees of freedom $\Phi$ and $\Psi$, as guaranteed by Theorem \ref{theo:cosmo2pert}.

Now, write the gravitational equations of motion in the form
\begin{equation}
\mathcal{E}_{ab}=8 \pi G T_{ab}\,,
\end{equation}
where $T_{ab}$ stands for the stress-energy tensor, which one could take to be that of a perfect fluid, if desired. Let $\mathcal{E}_{ab}^{(n)}$ denote the contribution to $\mathcal{E}_{ab}$ arising from the $n$-th order Cosmological GQTG $\mathcal{Q}_{(n)}$. Define the Hubble factor $H$ in terms of the comoving time $t$, $a H= \dot{a}$. In terms of the conformal time $\eta$, we trivially have $a^2 H=a'$, where \emph{prime} denotes the partial derivative with respect to conformal time $\eta$. Write down the perturbations in Fourier space:
\begin{equation}
\Psi=\int\diff^3k\, \Psi_{\mathbf{k}}\ex{\iu \mathbf{k}\cdot\mathbf{x}}\, ,\quad \Phi=\int\diff^3k\, \Phi_{\mathbf{k}}\ex{\iu \mathbf{k} \cdot\mathbf{x}}\,,
\end{equation}
for three-momentum $\mathbf{k}=(k_x,k_y,k_z)$. Let $\tilde{\E}_{ab}$ and $\tilde{\E}^{(n)}_{ab}$ be the subsequent equations of motion in Fourier space: 
\begin{equation}
\tilde{\E}_{ab}=\frac{1}{(2\pi)^3} \int \diff^3 x \, \E_{ab}\, \ex{- \iu \mathbf{k}\cdot\mathbf{x}}  \,,\quad  \tilde{\E}_{ab}^{(n)}=\frac{1}{(2\pi)^3} \int \diff^3 x \, \E_{ab}^{(n)}\, \ex{- \iu \mathbf{k}\cdot\mathbf{x}}  \,,
\end{equation}
Inspired by previous work \cite{Cano:2020oaa}, we decompose $\tilde{\E}^{(n)}_{ab}$ as follows:
\begin{equation}
\tilde{\E}_{\eta\eta}^{(n)}\, ,\quad \tilde{\E}_{\eta i}^{(n)}=\iu k_i A^{(n)}\, ,\quad  \tilde{\E}_{ij}^{(n)}=k_ik_j B^{(n)}+C^{(n)}\gamma_{ij}\, ,
\end{equation}
where $i$ and $j$ refer to the spatial indices. Then, we obtain the following expressions for $\tilde{\E}_{\eta\eta}^{(n)}$, $A^{(n)}$, $B^{(n)}$ and $C^{(n)}$  at linear order in the perturbations and for the Cosmological GQTGs of curvature orders  $n=3,4$ and $5$.


\subsubsection*{Cubic case}

The equations of motion for the scalar perturbations in the cubic theory $\mathcal{Q}_{(3)}$ read as follows:
\begin{small}
\begin{align}
\frac{\E_{\eta\eta}^{(3)}}{432}&=\Psi_\mathbf{k} \left(\frac{k^4 H'}{a^3}+2 k^2 H^4\right)+4 a^2 H^6 \Phi_\mathbf{k}+6 a H^5 \Psi_\mathbf{k}'\, ,\\
\frac{A^{(3)}}{432}&=-\frac{k^2 \Psi_\mathbf{k} H''}{a^3}-\frac{H \Phi_\mathbf{k}}{a^2} \left(2 a^3 H^4+k^2 H'\right)-\frac{\Psi_\mathbf{k}'}{a^3} \left(2
   a^3 H^4+k^2 H'\right)\, ,\label{eq:Eetai3}\\
\frac{B^{(3)}}{216}&=-\frac{3 H H' \Phi_\mathbf{k}'}{a^2}-\frac{3 H' \Psi_\mathbf{k}''}{a^3}+\frac{\Phi_\mathbf{k}}{a^3} \left(-3 a H H''-3 a
   H'^2-2 a^3 H^4-k^2 H'\right)\notag\\
&+\frac{\Psi_\mathbf{k}'}{a^3} \left(3 a H H'-6 H''\right)+\frac{\Psi_\mathbf{k}}{a^3} \left(3 a H
   H''+8 a^2 H^2 H'+2 a^3 H^4-3 H^{(3)}-2 k^2 H'\right)\, , \label{eq:Eij3}\\
\frac{C^{(3)}}{216}&=\Phi_\mathbf{k}' \left(\frac{k^2 H H'}{a^2}-4 a H^5\right)+\Psi_\mathbf{k}'' \left(\frac{k^2 H'}{a^3}-4 H^4\right)+\Psi_\mathbf{k}' \left(\frac{2 k^2 H''}{a^3}-\frac{k^2 H H'}{a^2}-8 a H^5-16 H^3 H'\right)\notag\\
&+\Phi_\mathbf{k} \left(\frac{k^2 H H''}{a^2}+2 H^4 \left(k^2-12 a H'\right)+\frac{k^2 H' \left(a H'+k^2\right)}{a^3}-12 a^2 H^6\right)\notag\\
&-\frac{\Psi_\mathbf{k}}{a^3} \left(k^2 a H H''+8 k^2 a^2
   H^2 H'+8 a^4 H^4 H'+2 k^2 a^3 H^4+4 a^5 H^6-k^2 H^{(3)}-2 k^4 H'\right)\, .\label{eq:Exx3}
\end{align}
\end{small}
Scalar perturbations of Cosmological GQTGs at cubic order have already been explored in the literature \cite{Cisterna:2018tgx,Cano:2020oaa}. Since, as commented before, there is a unique (up to an innocent global factor) four-dimensional Cosmological GQTG at cubic order, the equations we have just found should exactly coincide with those of \cite{Cisterna:2018tgx,Cano:2020oaa}. This is the case for the equations in Reference \cite{Cano:2020oaa}, but some disagreements with the values presented in Reference \cite{Cisterna:2018tgx} have been encountered\footnote{Note that the equations of \cite{Cisterna:2018tgx} are written in Cartesian coordinates and with the Hubble factor directly in terms of the conformal time, so some massaging needs to be carried out before direct comparison.}. Given the exact matching of our equations with those of \cite{Cano:2020oaa}, we believe the disagreements with \cite{Cisterna:2018tgx} are due to the present of some small typos in the latter Reference. 




\subsubsection*{Quartic case}
In this case, the linearized equations of motion read

\begin{small}
\begin{align}
\frac{\E_{\eta\eta}^{(4)}}{864}&= \Phi_\mathbf{k} \left(72 a^2 H^8-\frac{k^4 H'^2}{a^4}\right)+\Psi_\mathbf{k} \left(\frac{24 k^4 H^2 H'}{a^3}+\frac{7 k^4 H'^2}{a^4}+32 k^2 H^6\right)+96 a H^7 \Psi_\mathbf{k}'\, , \\
\frac{A^{(4)}}{864}&=\frac{k^2 H'^2 \Phi_\mathbf{k}'}{a^4}+\Psi_\mathbf{k}' \left(-\frac{24 k^2 H^2 H'}{a^3}-\frac{7 k^2 H'^2}{a^4}-32
   H^6\right)-\frac{2 k^2
   \Psi_\mathbf{k}}{a^4} \left(3 a H \left(4 H H''+3 H'^2\right)+7 H' H''\right)\notag\\
   &+\frac{2 \Phi_\mathbf{k}}{a^4}\left(k^2 H' H''-a H \left(12 k^2 a H^2 H'+16 a^4 H^6+5 k^2 H'^2\right)\right)\, ,\\
\frac{B^{(4)}}{432}&=\frac{3 H'^2 \Phi_\mathbf{k}''}{a^4}-\frac{3 H' \Psi_\mathbf{k}''}{a^4} \left(24 a H^2+7 H'\right)+\frac{12 H'\Phi_\mathbf{k}'}{a^4} \left(H''-3 a H \left(2 a H^2+H'\right)\right)\notag\\
&-\frac{12 \Psi_\mathbf{k}'}{a^4} \left(-6 a^2 H^3 H'+a H \left(12 H H''+13 H'^2\right)+7 H' H''\right)-\frac{\Phi_\mathbf{k}}{a^4} \Big(6 a H' \left(4 H \left(3 H''+k^2 H\right)+5 H'^2\right)\notag\\
&\left.+6 a^2
   H^2 \left(12 H H''+31 H'^2\right)+32 a^4 H^6-6 H''^2+7 k^2 H'^2-6 H^{(3)} H'\right)-\frac{\Psi_\mathbf{k}}{a^4} \left(-192 a^3 H^4
   H'\right.\notag\\
   &\left.-18 a^2 H^2 \left(4 H H''+3 H'^2\right)+6 a \left(9 H'^3+4 H^2 \left(3 H^{(3)}+2 k^2 H'\right)+28 H H' H''\right)\right.\notag\\
   &\left.-32 a^4 H^6+42 H''^2-9 k^2 H'^2+42 H^{(3)} H'\right)\, ,\\ \notag
\frac{C^{(4)}}{432}&=\frac{4 \Phi_\mathbf{k}'}{a^4} \left(6 k^2 a^2 H^3 H'+3 k^2 a H H'^2-16 a^5 H^7-k^2 H' H''\right)-\frac{k^2 H'^2 \Phi_\mathbf{k}''}{a^4}+k^2 \Psi_\mathbf{k}'' \left(\frac{24  H^2 H'}{a^3}+\frac{7  H'^2}{a^4}\right)\\\notag &-64H^6 \Psi_\mathbf{k}''
   +\frac{4 \Psi_\mathbf{k}'}{a^4} \left(a H \left(13 k^2 H'^2-6 a H^2 H' \left(16 a^2 H^2+k^2\right)-32 a^4 H^6\right)+k^2 H'' \left(12 a H^2+7
   H'\right)\right)\\\notag &+\Phi_\mathbf{k} \left(-512 a H^6 H'+\frac{2 k^2 H^2 \left(12 H H''+31 H'^2\right)}{a^2}+\frac{2 k^2 H' \left(12
   H \left(H''+k^2 H\right)+5 H'^2\right)}{a^3}\right) \\& \notag+\Phi_{\mathbf{k}} \left (\frac{k^2 \left(-2 H''^2+7 k^2 H'^2-2 H^{(3)} H'\right)}{a^4}-192 a^2 H^8+32 k^2
   H^6\right)-\Psi_\mathbf{k} H^4 H' \left(\frac{192 k^2}{a}+128 a H^2  \right) \\ \notag &+ \Psi_{\mathbf{k}} \left (-\frac{6 k^2 H^2 \left(4 H H''+3 H'^2\right)}{a^2}+\frac{2 k^2
   \left(12 H^2 H^{(3)}+9 H'^3+4 H H' \left(7 H''+6 k^2 H\right)\right)}{a^3} \right)\\&  +\Psi_{\mathbf{k}} \left (\frac{k^2 \left(14 H''^2-9 k^2 H'^2+14 H^{(3)}
   H'\right)}{a^4}-48 a^2 H^8-32 k^2 H^6\right)\,.
\end{align}
\end{small}

If we compare these equations to those of Reference \cite{Cano:2020oaa}, we observe that they are not the same. But this is of course to be expected, since our theory and the one considered by that Reference are not the same, see Equation \eqref{eq:pabloynuestra}. Nonetheless, we have been able to check that on adding the quartic trivial GQTG introduced in Equation \eqref{eq:trivialnuestra}, we recover exactly (up to an innocent global factor) the equations of Reference \cite{Cano:2020oaa}.

\section*{Quintic case}

Finally, we present the perturbations for the $n=5$ order. The expression of the associated equations of motion are given by
\begin{small}
\begin{align}
\notag
\frac{\E_{\eta\eta}^{(5)}}{1728}&=\frac{\Phi_\mathbf{k}}{a^5} \left(-30 k^4 a H^2 H'^2+864 a^7 H^{10}-17 k^4 H'^3\right)+\Psi_\mathbf{k} \left(\frac{360 k^4 H^4
   H'}{a^3}+\frac{210 k^4 H^2 H'^2}{a^4}+\frac{37 k^4 H'^3}{a^5} \right) \\& +360 k^2 H^8\Psi_{\mathbf{k}}+1080 a H^9 \Psi_\mathbf{k}'\, , \\ \notag
\frac{A^{(5)}}{1728}&=\Phi_\mathbf{k}' \left(\frac{17 k^2 H'^3}{a^5}+\frac{30 k^2 H^2 H'^2}{a^4}\right)+\Psi_\mathbf{k}' \left(-\frac{360 k^2 H^4
   H'}{a^3}-\frac{210 k^2 H^2 H'^2}{a^4}-\frac{37 k^2 H'^3}{a^5}-360 H^8\right)\\ \notag& +\frac{3 \Phi_\mathbf{k}}{a^5} \left(k^2 H' H'' \left(20 a
   H^2+17 H'\right)-5 a H \left(24 k^2 a^2 H^4 H'+20 k^2 a H^2 H'^2+24 a^5 H^8+3 k^2 H'^3\right)\right)\\&-\frac{3 k^2 \Psi_\mathbf{k}}{a^5} \left(60 a^2 H^3 \left(2 H H''+3 H'^2\right)+5 a H H' \left(28 H H''+13 H'^2\right)+37 H'^2 H''\right)\, ,\\ \notag
\frac{B^{(5)}}{864}&=\frac{3 H'^2 \Phi_\mathbf{k}''}{a^5} \left(30 a H^2+17 H'\right)+\frac{\Psi_\mathbf{k}''}{a^5} \left(-1080 a^2 H^4 H'-630 a
   H^2 H'^2-111 H'^3\right) \\\notag& -\frac{18 H' \Phi_\mathbf{k}'}{a^5} \left(6 a H \left(10 a H^2 H'+10 a^2 H^4+H'^2\right)-H'' \left(20
   a H^2+17 H'\right)\right) \\ \notag &  +\frac{18 \Psi_\mathbf{k}'}{a^5} \left(60 a^3 H^5 H'-20 a^2 H^3 \left(6 H H''+13 H'^2\right)-28 a H
   H' \left(5 H H''+3 H'^2\right)-37 H'^2 H''\right)\\\notag &  +\frac{\Phi_\mathbf{k}}{a^5} \left(-90 a^2 H^2 H' \left(4 H \left(6 H''+k^2
   H\right)+27 H'^2\right)-180 a^3 H^4 \left(6 H H''+25 H'^2\right) \right)\\\notag &  +\frac{\Phi_\mathbf{k}}{a^5}\left (-3 a \left(45 H'^4+168 H H'^2 H''+10 H^2 \left(-6 H''^2+7 k^2
   H'^2-6 H^{(3)} H'\right)\right)-360 a^5 H^8 \right) \\\notag &  +\frac{\Phi_{\mathbf{k}}}{a^5} \left(H'\left(306 H''^2-37 k^2 H'^2+153 H^{(3)} H'\right)\right)+\frac{\Psi_\mathbf{k}}{a^5} 
   \left(2880 a^4 H^6 H'+540 a^3 H^4 \left(2 H H''+3 H'^2\right)\right)\\\notag & +\frac{\Psi_{\mathbf{k}}}{a^5}\left(-90 a^2 H^2 \left(12 H^2 H^{(3)}+41 H'^3+8 H H' \left(7 H''+k^2
   H\right)\right)+9 a \left(-65 H'^4-364 H H'^2 H''\right)\right) \\\notag & +\frac{\Psi_{\mathbf{k}}}{a^5} \left( 90a H^2 \left(3 k^2 H'^2-14 H''^2-14 H^{(3)} H'\right)+H'
   \left(-666 H''^2+77 k^2 H'^2-333 H^{(3)} H'\right)\right) \\& +360  H^8 \Psi_{\mathbf{k}}\, ,\\ \notag
\frac{C^{(5)}}{864}&=-\frac{k^2 H'^2 \Phi_\mathbf{k}''}{a^5} \left(30 a H^2+17 H'\right)+\Psi_\mathbf{k}'' \left(\frac{360 k^2 H^4
   H'}{a^3}+\frac{210 k^2 H^2 H'^2}{a^4}+\frac{37 k^2 H'^3}{a^5}-720 H^8\right)\\& \notag +\frac{6 \Phi_\mathbf{k}'}{a^5} \left(6 a H \left(10 k^2 a^2
   H^4 H'+10 k^2 a H^2 H'^2-20 a^5 H^8+k^2 H'^3\right)-k^2 H' H'' \left(20 a H^2+17 H'\right)\right)\\ \notag &-\frac{6 \Psi_\mathbf{k}'}{a^5} \left(4 a H \left(15 a^2 H^4 H' \left(16 a^2 H^2+k^2\right)-65 k^2 a H^2 H'^2+60 a^5 H^8-21 k^2 H'^3\right)\right)\\\notag &+\frac{ 6 k^2 H''\Psi'_{\mathbf{k}}}{a^5}
   \left(140 a H^2 H'+120 a^2 H^4+37 H'^2\right)-60H^4\Phi_\mathbf{k} \left(120 a H^4 H'-\frac{k^2 \left(6 H H''+25   H'^2\right)}{a^2} \right)\\ \notag &+\Phi_\mathbf{k} \left (\frac{90 k^2 H^2 H' \left(8 H H''+9 H'^2+4 k^2 H^2\right)}{a^3}+\frac{k^2 H' \left(-102 H''^2+37 k^2 H'^2-51
   H^{(3)} H'\right)}{a^5}\right) \\ \notag & +\Phi_\mathbf{k} \left( \frac{3 k^2 \left(15 H'^4+56 H H'^2 H''+10 H^2 \left(7 k^2 H'^2-2 H''^2-2 H^{(3)}
   H'\right)\right)}{a^4}-2160 a^2 H^{10}+360 k^2 H^8\right)\\ \notag & +\Psi_\mathbf{k} \left(\frac{30 k^2 H^2 \left(12 H^2 H^{(3)}+41 H'^3+8 H H' \left(7 H''+3 k^2
   H\right)\right)}{a^3}-\frac{2880 k^2 H^6 H'}{a}-1440 a H^8 H' \right) \\ \notag &+ \Psi_{\mathbf{k}} \left (-\frac{180 k^2
   H^4 \left(2 H H''+3 H'^2\right)}{a^2}+\frac{k^2 H' \left(222 H''^2-77 k^2 H'^2+111 H^{(3)} H'\right)}{a^5}\right) \\ \notag & +\Psi_{\mathbf{k}}\left (\frac{3 k^2 \left(65 H'^4+364 H H'^2 H''+10
   H^2 \left(14 H''^2-9 k^2 H'^2+14 H^{(3)} H'\right)\right)}{a^4}-432 a^2 H^{10}\right)\\& -360 k^2 H^8 \Psi_{\mathbf{k}}\, .
\end{align}
\end{small}
We would like to emphasize that these represent the first computations ever of cosmological perturbations in quintic Cosmological GQTGs. 

\section{Conclusions}

\label{sec:conclu}

In conclusion, in this work we have carried out an exhaustive study of higher-curvature gravities with second-order equations of motion on FLRW backgrounds in generic dimensions $D \geq 3$. First, after examining the properties of curvature invariants when evaluated on FLRW backgrounds, we were able to derive Proposition \ref{prop:cosmocond}, which provides an extremely simple requirement to check whether a given higher-curvature gravity is of the cosmological type in generic $D$. This condition depends solely on the properties of the theory when evaluated on FLRW ans\"atze. Next, we obtained the unique expression, in terms of the Ricci scalar and the unique independent component of the traceless Ricci tensor on FLRW configurations, which satisfies the aforementioned condition at each curvature order and dimension $D$. With this result, we obtained an instance of Cosmological Gravity at each curvature order and spacetime dimension $D$. Since one is free to add and remove terms that vanish on top of FLRW backgrounds (such as any curvature invariant containing at least one Weyl curvature tensor), our result completes the classification of all inequivalent Cosmological Gravities, defined as those equivalence classes arising after stating that two Cosmological Gravities are equivalent if they differ by a piece that vanishes for FLRW backgrounds. In the literature, examples of Cosmological Gravities had been built at all curvature orders in $D=3$ and a recursive relation was derived to obtain theories (with covariant derivatives of the curvature) of the cosmological type at all orders and dimensions (although studied from the holographic $c$-theorem perspective), so our results complete the study of Cosmological Gravities without covariant derivatives of the curvatures for all orders and dimensions, providing explicit and relatively simple examples of such theories. In the case of $D=4$, we presented some physical applications of our Cosmological Gravities and noted that they naturally feature the phenomenon of geometric inflation \cite{Cisterna:2018tgx,Arciniega:2018tnn}, by which the Universe undergoes an inflationary era --- with the appropriate number of e-folds --- induced by the presence of higher-curvature terms, without the introduction of additional matter fields.





Next, we studied the compatibility between the cosmological condition and the requirement of being a GQTG. We considered $D \geq 4$ and we observed that not every Cosmological Gravity belongs to the GQTG class. Nevertheless, we were able to prove two remarkable results. On the one hand, for dimensions $D \geq 5$, we showed that Quasitopological Gravities (i.e., GQTGs with algebraic equations of motion for static and spherically symmetric configurations) can always be transformed into a Cosmological Quasitopological Gravity by adding a prescribed combination of curvature invariants that vanishes on top of static and spherically symmetric backgrounds. We proved this by showing an example of such a theory at all curvature orders and dimensions $D \geq 5$. On the other hand, for $D=4$, we presented an instance of a Cosmological GQTG at all curvature orders, thus completing previous works in the literature which found examples of Cosmological GQTGs up to eighth order in the curvature. Observe that all the aforementioned Cosmological Gravities (not necessarily fulfilling the GQTG condition) will also satisfy a holographic $c$-theorem by construction \cite{Bueno:2022log}.

Finally, we studied the properties of cosmological perturbations in Cosmological Gravities. Remarkably, we showed that every Cosmological Gravity for any dimension $D\geq 3$ is such that the subsequent linearized equations of motion for scalar cosmological perturbations possess at most two time derivatives (although they may contain spatial derivatives of higher order). Then we restricted our attention to the four-dimensional Cosmological GQGTs we had previously derived, provided the necessary ingredients to compute the linear equations of motion for scalar cosmological perturbations and presented these equations for theories up to fifth order in the curvature. 

In the future, there are several intriguing avenues to be explored. For example, it would be interesting to further characterize the class of Cosmological Gravities in generic dimensions. Is it possible to obtain full classification results of these theories at all curvature orders, beyond inequivalent Cosmological Gravities? Similarly, one could ask the same question in the context of Cosmological GQTGs. More specifically, in the same way it is possible to transform Quasitopological Gravities in $D \geq 5$ and at any curvature order into Cosmological Quasitopological Gravities by adding combinations of curvature invariants that vanish on static and spherically symmetric configurations, could we upgrade proper GQTGs in $D \geq 5$ into Cosmological GQTGs\footnote{We were able to check this for a quartic five-dimensional proper GQTG, see Equation \eqref{eq:propgqg5}. In another vein, in the four-dimensional case we showed how to convert proper GQTGs at any curvature order into Cosmological GQTGs.}? On the other hand, in the four-dimensional case, it is tantalizing to find the expression for all theories that fulfill both the definitions of Cosmological Gravity and GQTG. In particular, one may wonder about the properties of these theories beyond FLRW and static and spherically symmetric configurations. For instance, it can be checked that the unique four-dimensional cubic Cosmological GQTG also possesses Taub-NUT solutions whose equations of motion are of second order in derivatives \cite{Bueno:2018uoy}, so it would be interesting to investigate whether this trend continues to all curvature orders or whether there are further constraints to be imposed.


On a different front, another clear future direction corresponds to the study of vector and tensor perturbations in the context of Cosmological Gravities. As mentioned in Section \ref{sec:cosmopert}, while scalar perturbations are guaranteed to have equations of motion with no more than two time derivatives, this is no longer the case for tensor perturbations in Cosmological Gravities. However, there are specific Cosmological Gravities in which cosmological tensor perturbations do obey equations of motion with no more than two temporal derivatives \cite{Cano:2020oaa}, so it is intriguing to examine which further conditions need to be required so that cosmological tensor (and vector) perturbations possess equations of second order in time derivatives. 

Finally, there are two canonical extensions of our Cosmological Gravities whose examination would be interesting in the future. First, one could investigate the addition of (non-minimally coupled) matter to Cosmological Gravities, such as a scalar field or a  $\mathrm{U}(1)$ gauge field. Regarding this last point, it is known that the definition of GQTGs may be canonically extended to allow for non-minimal couplings to a $\mathrm{U}(1)$ vector \cite{Cano:2020ezi,Cano:2020qhy,Bueno:2021krl,Cano:2022ord,Bueno:2022ewf}, so it makes sense to explore whether a similar situation may occur for Cosmological Gravities, Also, the inclusion of this type of matter has mainly been considered in the context of minimally coupled non-linear electrodynamics \cite{Garcia-Salcedo:2000ujn,PhysRevD.65.063501,Novello:2008ra}, so this aspect deserves careful scrutiny. Second, a natural question would be the addition of terms with covariant derivatives of the curvature, expected to appear from an effective field theory perspective.  In particular, after the recent discovery of GQTGs with covariant derivatives \cite{Aguilar-Gutierrez:2023kfn}, it is tantalizing to investigate whether there exist Cosmological Gravities featuring covariant derivatives of the curvature which also satisfy the GQTG condition. We expect to address these directions in the near future.


\section*{Acknowledgments}
We thank Daniele Bertacca, Pablo Bueno, Pablo A. Cano, Nicolás Grandi and Julio Oliva for useful discussions. The research of J. M. has been supported by Israel Science Foundation, grant no. 1487/21 and by FONDECYT Postdoctorado Grant 3230626.  The work of Á. J. M. has been supported by a postdoctoral grant from the Istituto Nazionale di Fisica Nucleare, Bando 23590. Á. J. M. also wishes to thank E. Gil and A. Murcia for their permanent support.

\appendix



\section{Cosmological Quasitopological Gravities in $D \geq 5$ at orders $n=5$ and $6$}\label{App:A}

We present here examples of Cosmological Quasitopological Gravities of order 5 and 6 in the curvature in $D \geq 5$. They have been obtained by setting $n=5$ and $6$ in Equation \eqref{eq:quasicosmo}.

\begin{align}
\nonumber
\mathsf{Q}_{(5)}&=R^5+\frac{4 (D-1)^3 D^4 (4 D-5) W_{ghij} W^{ghij}
   W\indices{_a_b^c^d}W\indices{_c_d^e^f}W\indices{_e_f^a^b}}{(D-3)^2 (D-2)^3 (D ((D-9)
   D+26)-22)}\\\nonumber &+\frac{48 (D-1)^3 D^4 \left(16 (D-3)
   (D-1) Z_{a}^b  Z_{b}^{c} Z_{c}^d Z_{d}^{f} Z_{f}^a-5 (D-2)^2 W_{ghij} W^{ghij} Z^a_b W_{a c d e}W^{bcde}\right)}{(D-4)
   (D-3)^2 (D-2)^6}\\\nonumber &-\frac{640 (D-1)^4 D^4 Z_{de} Z^{de}
   Z^a_b Z^b_cZ_a^c}{(D-4) (D-3) (D-2)^6}+\frac{15 (D-1)^2
   D^3 (3 D-4) R \left (W_{abcd} W^{abcd} \right)^2}{(D-3)^2
   (D-2)^4}\\\nonumber &-\frac{1920 (D-1)^3 D^3 R Z^{a}_b Z_{a c} Z_{d e} W^{bdce}}{(D-4)
   (D-3) (D-2)^4}+\frac{240 (D-1)^3 D^3 R
   \left ( Z_{ab} Z^{ab} \right)^2}{(D-3) (D-2)^5}\\\nonumber &-\frac{480 (D-1)^3 D^3 R
  Z^a_bZ^b_cZ^c_dZ^d_a}{(D-3) (D-2)^5}-\frac{160 (D-1)^3 D^3
   (D (11 D-12)+4) W_{ghij} W^{ghij} Z^a_b Z^b_cZ_a^c}{(D-4) (D-3)
   (D-2)^6}\\\nonumber & +\frac{40 (D-1)^3 D^3 (D (17
   D-28)+12) W\indices{_a_b^c^d}W\indices{_c_d^e^f}W\indices{_e_f^a^b} Z_{gh} Z^{gh}}{(D-3) (D-2)^4 (D
   ((D-9) D+26)-22)}\\\nonumber & +\frac{1920 (D-1)^4 D^3 Z_{ef} Z^{ef}
   Z^a_b Z^c_d W\indices{_a_c^b^d}}{(D-3) (D-2)^6} +\frac{20 (D-1)^2 D^2 (2
   D-3) R^2 W\indices{_a_b^c^d}W\indices{_c_d^e^f}W\indices{_e_f^a^b}}{(D-3) (D-2) (D ((D-9)
   D+26)-22)}\\\nonumber & +\frac{240 (D-1)^2 D^2 R^2
   Z^a_b Z^c_d W\indices{_a_c^b^d}}{(D-3) (D-2)^3}+\frac{160 (D-1)^2 D^2 R^2
   Z^a_b Z^b_cZ_a^c}{(D-2)^4}\\\nonumber & -\frac{960 (D-1)^3 D^2 R Z^{ab} W_{acbd} W^{c efg} W^d{}_{efg}}{(D-4)
   \left(D^2-5 D+6\right)^2} -\frac{240 (D-1)^2 D^2 R^2
   Z^a_b W_{a c d e}W^{bcde}}{(D-4) (D-3) (D-2)^2}\\\nonumber & +\frac{120 (D-1)^2
   D^2 (D (7 D-10)+4) R W_{abcd} W^{abcd} Z_{ef}Z^{ef}}{(D-3)
   (D-2)^5}+\frac{10 (D-1) D R^3
   W_{abcd} W^{abcd}}{(D-3) (D-2)}\\&-\frac{40 (D-1) D R^3
   Z_{ab} Z^{ab}}{(D-2)^2}\,,\\
   \nonumber
\mathsf{Q}_{(6)}&=R^6+\frac{15 (D-1) D W_{abcd} W^{abcd} R^4}{(D-3) (D-2)}-\frac{60 (D-1) D Z_{ab} Z^{ab} R^4}{(D-2)^2}\\\nonumber &+\frac{40 (D-1)^2 D^2 (2 D-3)
   W\indices{_a_b^c^d}W\indices{_c_d^e^f}W\indices{_e_f^a^b} R^3}{(D-3) (D-2) (D ((D-9) D+26)-22)}+\frac{480 (D-1)^2 D^2 Z^a_b Z^c_d W\indices{_a_c^b^d} R^3}{(D-3) (D-2)^3}\\\nonumber & +\frac{320 (D-1)^2
   D^2 Z^a_b Z^b_cZ_a^c R^3}{(D-2)^4}-\frac{480 (D-1)^2 D^2 Z^a_b W_{a c d e}W^{bcde} R^3}{(D-4) (D-3) (D-2)^2}\\\nonumber &+\frac{45 (D-1)^2 D^3 (3 D-4)
   \left(W_{a b c d} W^{a b c d}\right)^2 R^2}{(D-3)^2 (D-2)^4}+\frac{720 (D-1)^3 D^3 \left(Z^a_b Z_a^b\right)^2 R^2}{(D-3) (D-2)^5}\\\nonumber &-\frac{5760 (D-1)^3 D^3 Z^{a}_b Z_{a c} Z_{d e} W^{bdce}
   R^2}{(D-4) (D-3) (D-2)^4}+\frac{360 (D-1)^2 D^2 (D (7 D-10)+4) W_{abcd} W^{abcd} Z_{ef} Z^{ef} R^2}{(D-3) (D-2)^5}\\\nonumber &-\frac{1440 (D-1)^3
   D^3 Z^a_bZ^b_cZ^c_dZ^d_a R^2}{(D-3) (D-2)^5}-\frac{2880 (D-1)^3 D^2 Z^{ab} W_{acbd} W^{c efg} W^d{}_{efg} R^2}{(D-4) \left(D^2-5 D+6\right)^2}\\\nonumber &+\frac{24 (D-1)^3
   D^4 (4 D-5) W_{a b c d} W^{a b c d}W\indices{_e_f^g^h}W\indices{_g_h^i^j}W\indices{_i_j^e^f} R}{(D-3)^2 (D-2)^3 (D ((D-9) D+26)-22)}\\\nonumber &+\frac{240 (D-1)^3 D^3 (D (17 D-28)+12)
   W\indices{_a_b^c^d}W\indices{_c_d^e^f}W\indices{_e_f^a^b}Z^g_hZ^h_g R}{(D-3) (D-2)^4 (D ((D-9) D+26)-22)}\\\nonumber &+\frac{11520 (D-1)^4 D^3 Z^a_b Z_a^bZ^c_d Z^e_f W\indices{_c_e^d^f} R}{(D-3) (D-2)^6}-\frac{960
   (D-1)^3 D^3 (D (11 D-12)+4) W_{a b c d} W^{a b c d}Z^e_f Z^f_gZ_e^g R}{(D-4) (D-3) (D-2)^6}\\\nonumber &-\frac{3840 (D-1)^4 D^4 Z_{ab} Z^{ab} Z_{ef} Z^{fg} Z_{g}{}^e
   R}{(D-4) (D-3) (D-2)^6}+\frac{4608 (D-1)^4 D^4 Z_{ab} Z^{bc} Z_{cd} Z^{de} Z_{e}{}^a R}{(D-4) (D-3) (D-2)^6}\\\nonumber &-\frac{1440 (D-1)^3 D^4  W_{a b c d} W^{a b c d}Z^e_f W_{e g h i}W^{fghi} R}{(D-4) (D-3)^2 (D-2)^4}+\frac{5 (D-1)^3 D^5 (5 D-6) \left ( W_{abcd} W^{abcd}\right)^3}{(D-3)^3 (D-2)^6}\\\nonumber &+\frac{320 (D-1)^4 D^3
   (D (9 D-26)+24) \left ( Z_{ab} Z^{ab} \right)^3}{(D-3) (D-2)^8}-\frac{15360 (D-1)^5 D^4
   Z^{a}_b Z_{a c} Z_{d e} W^{bdce} Z_{fg} Z^{fg}}{(D-4) (D-3) (D-2)^7}\\\nonumber & +\frac{960 (D-1)^4 D^3 (D (8 D-7)+2) W_{abcd} W^{abcd} \left ( Z_{ef} Z^{ef} \right)^2}{(D-3)
   (D-2)^8}\\\nonumber & +\frac{60 (D-1)^3 D^4 \left(31 D^2-54 D+24\right) \left(W_{a b c d} W^{a b c d}\right)^2 Z_{ef} Z^{ef}}{(D-3)^2 (D-2)^7}\\\nonumber &-\frac{320 (D-1)^4 D^4 (D (23 D-32)+12) W\indices{_a_b^c^d}W\indices{_c_d^e^f}W\indices{_e_f^a^b}Z^g_h Z^h_iZ_g^i}{(D-4) (D-3) (D-2)^5
   (D ((D-9) D+26)-22)}\\\nonumber &-\frac{640 (D-1)^4 D^4 (5 D-6) Z_{ef} Z^{ef} Z^a_bZ^b_cZ^c_dZ^d_a}{(D-3) (D-2)^8}\\&-\frac{960 (D-1)^4 D^4 W_{hijk} W^{hijk}
   Z^{ab} W_{acbd} W^{c efg} W^d{}_{efg}}{(D-4) (D-3)^3 (D-2)^4}\,.
\end{align}

\section{Four-dimensional Cosmological Generalized Quasitopological Gravities at orders $n=6,7,8$ and $9$}\label{App:B}

We devote this appendix to the presentation of examples of four-dimensional Cosmological Generalized Quasitopological Gravities from the sixth up to the ninth curvature order, thus extending previous results in the literature \cite{Arciniega:2018fxj,Cisterna:2018tgx,Arciniega:2018tnn}. These theories have been derived by setting $n=6$, $7$, $8$ and $9$ in Equation \eqref{eq:cosmogqgteo}. 
\begin{align}
\mathcal{Q}_{(6)}&=R^6+90 R^4 W_{a b c d} W^{a b c d}-180 R^4 Z^a_b Z_a^b-2520 R^3 W\indices{_a_b^c^d}W\indices{_c_d^e^f}W\indices{_e_f^a^b}+2160 R^3 Z^a_b Z^c_d W\indices{_a_c^b^d}\notag\\
&+2880 R^3 Z^a_b Z^b_cZ_a^c+2025 R^2 \left(W_{a b c d} W^{a b c d}\right)^2-3240 R^2 W_{a b c d} W^{a b c d}Z^e_f Z_e^f+4860 R^2
   \left(Z^a_b Z_a^b\right)^2\notag\\
   &-19440 R^2Z^a_bZ^b_cZ^c_dZ^d_a-8424 R W_{a b c d} W^{a b c d}W\indices{_e_f^g^h}W\indices{_g_h^i^j}W\indices{_i_j^e^f}\notag\\
   &+12960 R W\indices{_a_b^c^d}W\indices{_c_d^e^f}W\indices{_e_f^a^b}Z^g_hZ^h_g+20736 R Z^a_b Z_a^bZ^c_d Z^d_eZ_c^e+1080 \left(W_{a b c d} W^{a b c d}\right)^3\notag\\
   &-1620 \left(W_{a b c d} W^{a b c d}\right)^2Z^e_f Z_e^f+6480
   \left(Z^a_b Z_a^b\right)^3-25920 Z^a_b Z_a^b Z^c_d Z^d_e Z^e_f Z^f_c\, , \\   \nonumber
   \mathcal{Q}_{(7)}&=R^7+126 R^5 W_{abcd} W^{abcd}-252 R^5 Z_{ab} Z^{ab} -4410 R^4  W\indices{_a_b^c^d}W\indices{_c_d^e^f}W\indices{_e_f^a^b}\\\nonumber &+3780 R^4  Z^a_b Z^c_d W\indices{_a_c^b^d}+5040 R^4
   Z_{ab} Z^{ac} Z^{b}{}_c+4725 R^3 \left ( W_{abcd} W^{abcd}\right)^2\\ \nonumber &-7560 R^3 W_{abcd} W^{abcd} Z_{ef}Z^{ef}+11340 R^3
   \left ( Z_{ab} Z^{ab} \right)^2-45360 R^3 Z^a_bZ^b_cZ^c_dZ^d_a\\ \nonumber &-29484 R^2 W_{ghij} W^{ghij}  W\indices{_a_b^c^d}W\indices{_c_d^e^f}W\indices{_e_f^a^b}+45360 R^2  W\indices{_a_b^c^d}W\indices{_c_d^e^f}W\indices{_e_f^a^b}
   Z_{gh} Z^{gh}\\ \nonumber &+72576 R^2 Z_{ab} Z^{ab} Z^{cd} Z_{de} Z^{e}{}_c+7560 R \left (W_{abcd} W^{abcd} \right) ^3-11340 R \left ( W_{abcd} W^{abcd} \right) ^2
   Z_{ef}Z^{ef}\\ \nonumber &+45360 R \left ( Z_{ab} Z^{ab} \right) ^3-181440 R Z_{ef} Z^{ef} Z^a_bZ^b_cZ^c_dZ^d_a+62208 \left ( Z_{de} Z^{de} \right) ^2 Z^a_b Z^b_cZ_a^c\\&-9234 \left ( W_{ghij} W^{ghij} \right) ^2
    W\indices{_a_b^c^d}W\indices{_c_d^e^f}W\indices{_e_f^a^b}+13608 W_{ijkl} W^{ijkl}  W\indices{_a_b^c^d}W\indices{_c_d^e^f}W\indices{_e_f^a^b} Z_{gh} Z^{gh}\,,  \\ \nonumber 
\mathcal{Q}_{(8)}&=R^8+168 R^6 W_{abcd} W^{abcd}-336 R^6 Z_{a}^{b} Z^{a}_{b}-7056 R^5 W\indices{_a_b^c^d}W\indices{_c_d^e^f}W\indices{_e_f^a^b}+6048 R^5 Z^a_b Z^c_d W\indices{_a_c^b^d}\\ \nonumber &+8064 R^5
   Z^a_b Z^b_cZ_a^c+9450 R^4 \left (W_{abcd} W^{abcd} \right)^2-15120 R^4 W_{abcd} W^{abcd} Z_{ef} Z^{ef}\\ \nonumber & +22680 R^4
   \left(Z^a_b Z_a^b\right)^2-90720 R^4 Z^a_bZ^b_cZ^c_dZ^d_a-78624 R^3 W_{ghij} W^{ghij} W\indices{_a_b^c^d}W\indices{_c_d^e^f}W\indices{_e_f^a^b}\\ \nonumber &+120960 R^3 W\indices{_a_b^c^d}W\indices{_c_d^e^f}W\indices{_e_f^a^b}
   Z_{gh} Z^{gh}+193536 R^3 Z_{de} Z^{de} Z^a_b Z^b_cZ_a^c\\ \nonumber &+30240 R^2 \left ( W_{cdef} W^{cdef} \right)^3-45360 R^2 \left ( W_{cdef} W^{cdef} \right) ^2
   Z_{ab} Z^{ab}+181440 R^2 \left ( Z_{ab} Z^{ab} \right)^3\\ \nonumber &-725760 R^2 Z_{ef} Z^{ef} Z^a_bZ^b_cZ^c_dZ^d_a-73872 R \left ( W_{ghij} W^{ghij} \right) ^2
   W\indices{_a_b^c^d}W\indices{_c_d^e^f}W\indices{_e_f^a^b}\\  \nonumber &+108864 R W_{ijkl}W^{ijkl} W\indices{_a_b^c^d}W\indices{_c_d^e^f}W\indices{_e_f^a^b} Z_{gh} Z^{gh}+497664 R \left(Z^e_f Z_e^f\right)^2 Z^a_b Z^b_cZ_a^c\\ \nonumber &+6237
   \left ( W_{abcd} W^{abcd} \right )^4-9072\left ( W_{cdef} W^{cdef} \right)^3 Z_{ab} Z^{ab}+108864 \left ( Z_{ab} Z^{ab} \right) ^4\\ &-435456 \left(Z^e_f Z_e^f\right)^2
   Z^a_bZ^b_cZ^c_dZ^d_a\,, \\ \nonumber
\mathcal{Q}_{(9)}&=R^9 +216 R^7 W_{abcd} W^{abcd}-432 R^7 Z_{ab} Z^{ab}-10584 R^6 W\indices{_a_b^c^d}W\indices{_c_d^e^f}W\indices{_e_f^a^b}\\ \nonumber & +9072 R^6 Z^a_b Z^c_d W\indices{_a_c^b^d}+12096
   R^6 Z^a_b Z^b_cZ_a^c+17010 R^5 \left( W_{abcd} W^{abcd} \right)^2\\ \nonumber & -27216 R^5 W_{abcd} W^{abcd} Z_{ef} Z^{ef}+40824 R^5
   \left(Z^a_b Z_a^b\right)^2-163296 R^5 Z^a_bZ^b_cZ^c_dZ^d_a\\\nonumber & -176904 R^4 W_{ghij} W^{ghij} W\indices{_a_b^c^d}W\indices{_c_d^e^f}W\indices{_e_f^a^b}+272160 R^4 W\indices{_a_b^c^d}W\indices{_c_d^e^f}W\indices{_e_f^a^b}
   Z_{gh} Z^{gh}\\\nonumber & +435456 R^4 Z_{de} Z^{de} Z^a_b Z^b_cZ_a^c+90720 R^3 \left ( W_{abcd} W^{abcd} \right )^3\\\nonumber & -136080 R^3 \left ( W_{abcd} W^{abcd} \right )^2
   Z_{ef} Z^{ef}+544320 R^3 \left ( Z_{ab} Z^{ab} \right )^3\\\nonumber & -2177280 R^3 Z_{ef} Z^{ef} Z^a_bZ^b_cZ^c_dZ^d_a-332424 R^2
   \left ( W_{ghij} W^{ghij} \right )^2 W\indices{_a_b^c^d}W\indices{_c_d^e^f}W\indices{_e_f^a^b}\\\nonumber & +489888 R^2 W_{ghij} W^{ghij} W\indices{_a_b^c^d}W\indices{_c_d^e^f}W\indices{_e_f^a^b} Z_{kl} Z^{kl}+2239488 R^2
   \left(Z^e_f Z_e^f\right)^2 Z^a_b Z^b_cZ_a^c\\\nonumber & +56133 R \left ( W_{abcd} W^{abcd} \right )^4-81648 R \left ( W_{abcd} W^{abcd} \right )^3 Z_{ef} Z^{ef}+979776 R
   \left ( Z_{ab} Z^{ab} \right) ^4\\\nonumber & -3919104 R \left(Z^e_f Z_e^f\right)^2 Z^a_bZ^b_cZ^c_dZ^d_a-48600 \left ( W_{ghij} W^{ghij} \right )^3 W\indices{_a_b^c^d}W\indices{_c_d^e^f}W\indices{_e_f^a^b}\\& +69984
   \left ( W_{ghij} W^{ghij} \right )^2 W\indices{_a_b^c^d}W\indices{_c_d^e^f}W\indices{_e_f^a^b} Z_{kl} Z^{kl}+995328 \left ( Z_{de} Z^{de} \right) ^3 Z^a_b Z^b_cZ_a^c   \,.
\end{align}

\bibliographystyle{JHEP-2}
\bibliography{Gravities.bib}

\end{document}